\documentclass[letterpaper,11pt]{article}

\usepackage{style}
\usepackage{shortcuts}

\begin{document}
	
\begin{center}
	\begin{minipage}[H]{14.5cm} 
		
		\begin{center}
		{\huge \bf Near-Optimal Directed\\[.5ex]Low-Diameter Decompositions}
			\end{center}
		\vspace{1cm}
		
		{\large \textbf{Karl Bringmann}\footnote{This work is part of the project TIPEA that has received funding from the European Research Council (ERC) under the European Unions Horizon 2020 research and innovation programme (grant agreement No.~850979).} -- Saarland University and Max Planck Institute for Informatics, Saarland Informatics Campus, Germany} \vspace{1mm}\\
		{\large \textbf{Nick Fischer\footnote{Partially funded by the Ministry of Education and Science of Bulgaria's support for INSAIT as part of the Bulgarian National Roadmap for Research Infrastructure.}} -- INSAIT, Sofia University "St. Kliment Ohridski", Bulgaria} \vspace{1mm}\\
		{\large \textbf{Bernhard Haeupler\footnote{Partially funded by the Ministry of Education and Science of Bulgaria's support for INSAIT as part of the Bulgarian National Roadmap for Research Infrastructure and through the European Research Council (ERC) under the European Union's Horizon 2020 research and innovation program (ERC grant agreement 949272).}} -- INSAIT, Sofia University "St. Kliment Ohridski", Bulgaria \& ETH Zürich, Switzerland} \vspace{1mm}\\
		{\large \textbf{Rustam Latypov\footnote{Supported by the Research Council of Finland (grant 334238), The Finnish Foundation for Technology Promotion and The Nokia Foundation}} -- Aalto University, Finland} \vspace{1mm}\\

		\begin{abstract}
			\noindent
			Low Diameter Decompositions (LDDs) are invaluable tools in the design of combinatorial graph algorithms. While  historically they have been applied mainly to undirected graphs, in the recent breakthrough for the negative-length Single Source Shortest Path problem, Bernstein, Nanongkai, and Wulff-Nilsen [FOCS '22] extended the use of LDDs to directed graphs for the first time. Specifically, their LDD deletes each edge with probability at most \makebox{$O(\frac{1}{D} \cdot \log^2 n)$}, while ensuring that each strongly connected component in the remaining graph has a (weak) diameter of at most $D$.
			
			In this work, we make further advancements in the study of directed LDDs. We reveal a natural and intuitive (in hindsight) connection to Expander Decompositions, and leveraging this connection along with additional techniques, we establish the existence of an LDD with an edge-cutting probability of \makebox{$O(\frac{1}{D} \cdot \log n \log\log n)$}. This improves the previous bound by nearly a logarithmic factor and closely approaches the lower bound of \makebox{$\Omega(\frac{1}{D} \cdot \log n)$}. With significantly more technical effort, we also develop two efficient algorithms for computing our LDDs: a deterministic algorithm that runs in time~\smash{$\widetilde\Order(m \poly(D))$} and a randomized algorithm that runs in near-linear time~\smash{$\widetilde\Order(m)$}.

			We believe that our work provides a solid conceptual and technical foundation for future research relying on directed LDDs, which will undoubtedly follow soon.
		\end{abstract}
		
		\end{minipage}
	
\end{center}

\thispagestyle{empty}


\newpage
\pagenumbering{arabic}

\section{Introduction} \label{sec:intro}
In the design of combinatorial graph algorithms, decomposing graphs into smaller regions where problems become naturally easy to solve has emerged as a remarkably powerful paradigm. Two such graph decompositions stand out as particularly successful: \emph{Low-Diameter Decompositions} (LDDs) and \emph{Expander Decompositions}~(EDs) (defined formally in \Cref{def:ldd} and \Cref{sec:prelims}, respectively). Both decompositions remove some edges from the graph so that the remaining components either have a low diameter and thus behave nicely with respect to distance-type problems (in the case of LDDs), or are expanders which are well-suited for flow-type problems (in the case of EDs). This paradigm is particularly appealing as it typically leads to natural and efficient implementations in \emph{parallel,} \emph{distributed} and \emph{dynamic} models, and often also to \emph{deterministic} algorithms.

The vast majority of algorithmic results following this broader framework have been established for \emph{undirected} graphs. However, a recently increasing number of papers have successfully applied this paradigm also to \emph{directed} graphs~\cite{BernsteinGS20,BernsteinNW22,HuaKGW23,ChuzhoyK24,BernsteinBST24}, despite the serious technical difficulties that typically arise compared to the undirected setting. One celebrated example is the breakthrough \emph{near-linear time} algorithm for negative-length Single-Source Shortest Paths by Bernstein, Nanongkai and Wulff-Nilsen~\cite{BernsteinNW22}. In fact, their work was the first to define and apply Low-Diameter Decompositions in directed graphs.

Given these rapid and recent developments, it is to be expected that many new results will follow in this line of research. Consequently, anticipating a wave of researchers in need of directed Low-Diameter Decompositions, we initiate the systematic study of directed LDDs. Our main technical contribution is a \emph{near-optimal} LDD---near-optimal with respect to the loss parameter and the near-linear running time. Beyond improving the state of the art, we also unveil a compelling and novel connection between directed Low-Diameter Decompositions and Expander Decompositions, bringing together two previously studied concepts that are independently well-established. We believe this connection to be the main take-away of our paper, and are confident that it will inspire further applications. 

\subsection{Low-Diameter Decompositions} \label{sec:intro:sec:ldd}
Low-Diameter Decompositions in undirected graphs have been introduced almost 40 years ago by Awerbuch~\cite{Awerbuch85}, and, have since evolved into an irreplaceable tool in the design of combinatorial algorithms~\cite{AwerbuchGLP89,AwerbuchP92,AwerbuchBCP92,LinialS93,Bartal96,BlellochGKMPT14,MillerPX13,PachockiRSTW18,ForsterG19,ChechikZ20,BernsteinGW20,ForsterGV21,BernsteinNW22}. More generally, LDDs have been utilized in the development of diverse graph-theoretic structures such as oblivious routings~\cite{ZuzicGYHS22}, hopsets, neighborhood covers, and notably \emph{probabilistic tree embeddings}~\cite{Bartal96,Bartal98,FakcharoenpholRT04} (where the rough goal is to approximate an arbitrary (graph) metric by a simpler tree metric with polylogarithmic stretch) and \emph{low-stretch spanning trees}~\cite{AlonKPW95,ElkinEST08,AbrahamBN08,KoutisMP11,AbrahamN19,AbrahamCEFN20}. These structures in turn have led to further applications in approximation algorithms, online algorithms and network design problems~\cite{Bartal96,BorodinY98,HaeuplerHZ21}.

While the precise definitions of LDDs differed at first, all of them had in common that the graph was decomposed by removing few edges (on average) so that the remaining connected components have bounded diameter. More precisely, the modern definition is that an LDD is a probability distribution over edge cuts that cut each individual edge with probability at most $L/D$ (for some small loss factor~$L$) so that the resulting connected components have diameter at most $D$. For undirected graphs, the best-possible loss factor turns out to be $L = \Theta(\log n)$ with natural matching upper and lower bounds (see e.g.~\cite{Bartal96,FakcharoenpholRT04} and \cref{foot:logn-lower-bound}).

Only very recently, low-diameter decompositions debuted in \emph{directed} graphs. They played a key role in the near-linear time algorithm for the Single-Source Shortest Paths problem in graphs with negative edge lengths by Bernstein, Nanongkai and Wulff-Nilsen~\cite{BernsteinNW22} (following the paradigm outlined before). Curiously, before their seminal work, low-diameter decompositions for directed graphs were mostly unexplored---to the extent that it was even unclear what the definition should be. An even more recent application of directed LDDs is in the context of restricted (i.e., bicriteria) shortest paths~\cite{AshvinkumarBK25}.

\begin{restatable}[Directed Low-Diameter Decomposition]{definition}{defLDD} \label{def:ldd} 
	A \emph{directed low-diameter decomposition} with \emph{loss $L$} for a directed edge-weighted graph $G = (V, E, \ell)$ and a parameter $D \geq 1$ is a probability distribution over edge sets $S \subseteq E$ that satisfies the following two properties:
	\begin{itemize}
		\item For any two nodes $u, v \in V$ that are part of the same strongly connected component in $G \setminus S$, we have $d_G(u, v) \leq D$ and $d_G(v, u) \leq D$.
		\item For all edges $e \in E$, it holds that $\Pr(e \in S) \leq \frac{\ell(e)}{D} \cdot L$.
	\end{itemize}
\end{restatable}

That is, for directed graphs we require that all remaining \emph{strongly} connected components have diameter $D$ in the sense that each pair of nodes is at distance at most $D$ in the original uncut graph (this property is also called a ``weak'' diameter guarantee; see more details in the paragraph towards the end of \cref{sec:intro:sec:ldd}).

In light of this definition, it is natural to study directed LDDs with respect to two objectives: Minimizing the loss factor $L$, and \emph{computing} the LDD efficiently. Both of these requirements typically translate immediately to algorithmic improvements (e.g., the loss factor $L$ typically becomes a factor in the running time of algorithmic applications).

\paragraph{Quest 1: Minimizing the Loss \boldmath$L$.}
It is easy to see that a loss factor of $L \geq 1$ is necessary: Consider a graph consisting of disjoint copies of $D + 1$-cycles (of unit length, say). Any LDD is forced to cut at least one edge from every cycle, and thus some edges are deleted with probability~$\frac{1}{D+1}$. In fact, from the same lower bound as for undirected graphs one can show that $L \geq \Omega(\log n)$ is necessary\footnote{Specifically, using e.g., the probabilistic method, one can construct an undirected $n$-node graph $G$ with the following two properties: (i) the graph contains $c n$ edges (for some constant $c > 1$), and (ii) the girth of $G$ (i.e., the length of the shortest cycle) is at least $g = \Omega(\log n)$. Such a graph constitutes an obstruction against undirected LDDs with parameter $D = \frac{g}{2} - 1$, say. Indeed, all connected components that remain after the LDD cuts edges cannot contain cycles (as any cycle contains two nodes at distance at least $\frac{g}{2} > D$). This means that all connected components are trees, and thus the remaining graph is a forest containing at most $n - 1$ edges. Therefore, the LDD has cut at least~\makebox{$c n - n = \Omega(n)$} edges. In particular, the per-edge cutting probability of some edges must be~\makebox{$\Omega(1) = \Omega(\log n / D)$}. The same argument applies for directed LDDs by taking the same undirected graph and making all edges bidirected.\label{foot:logn-lower-bound}} for some directed graphs.

Conversely, Bernstein, Nanongkai and Wulff-Nilsen~\cite{BernsteinNW22} showed that a loss of $\Order(\log^2 n)$ can be achieved. Their work leaves open whether this factor can be improved, possibly to $\Order(\log n)$ which would match the existing unconditional lower bound, or whether $\Omega(\log^2 n)$ is necessary.

\paragraph{Quest 2: Efficient Algorithms.}
To be useful in algorithmic applications, it is necessary to be able to compute the LDD (or, formally speaking, to be able to \emph{sample} from the LDD efficiently). This is a non-negligible problem as many graph decompositions are much simpler to prove existentially than constructively and efficiently---for instance, for Expander Decompositions there is a stark contrast between the extremely simple existential proof and the involved efficient algorithms based on cut-matching games~\cite{RackeST14,KhandekarRaoVazirani09}. The $\Order(\log^2 n)$-loss due to~\cite{BernsteinNW22} is of course efficient.

\paragraph{Side Quest: Weak versus Strong.}
Another dimension is the distinction between~``weak'' and ``strong'' diameter guarantees. Specifically, \cref{def:ldd} requires the weak guarantee by requiring that any two nodes in the same strongly connected component are at distance at most $D$ \emph{in the original graph $G$}. The strong version instead requires that each strongly connected component in~$G \setminus S$ has diameter at most~$D$ \emph{in the graph $G \setminus S$.} While strong LDDs give a theoretically more appealing guarantee, for most algorithmic applications it turns out that weak LDDs suffice. The LDD developed by Bernstein, Nanongkai and Wulff-Nilsen~\cite{BernsteinNW22} has a weak guarantee, but follow-up work~\cite{BringmannCF23} later extended their results to a strong LDD with cubic loss $\Order(\log^3 n)$.

In this paper, we will mostly ignore the distinction between weak and strong and follow \cref{def:ldd} as it is. We remark however that all of our existential results do in fact have the strong diameter guarantee.
\section{Results and Technical Overview} \label{sec:overview}
In a nutshell, our result is that we establish a directed LDD with \emph{near-optimal} loss $\Order(\log n \log\log n)$. This comes tantalizingly close to the unconditional lower bound of $\Omega(\log n)$. It resembles a similar milestone in the development of probabilistic tree embeddings~\cite{Bartal98}, and also the current state of the art for low-stretch spanning trees~\cite{AbrahamN19}. In fact, similar $\log\log n$ barriers show up also in different contexts for directed graphs~\cite{ChechikLRS20}. In a sense, all of these results with $\log\log n$ overhead (including ours) apply a careful recursive approach that can be traced back to work by Seymour~\cite{Seymour95} (though with strongly varying implementations depending on the setting). 

We state our result in three \cref{thm:main-existential,thm:main-det,thm:main-fast}, where the first theorem is purely existential, the latter two are algorithmic, and the last has \emph{near-linear} running time.

\subsection{Near-Optimal LDDs via Expander Decompositions} \label{sec:overview:sec:ldd-exp}
The main conceptual contribution of our paper is that Low-Diameter Decompositions are closely and in a very practical way related to Expander Decompositions. Based on this insight, we prove the following theorem.

\begin{restatable}{theorem}{thmMainExistential} \label{thm:main-existential}
	For every directed graph there exists a directed LDD with loss $O(\log n \log \log n)$.
\end{restatable}

For the remainder of \cref{sec:overview:sec:ldd-exp} we will elaborate on the proof of \cref{thm:main-existential}. It involves two separate steps: reducing the problem to \emph{cost-minimizing} using the \emph{Multiplicative Weight Update} method, and then showing that a so called \emph{lopsided} expander decomposition is the desired \emph{cost-minimizer}. We emphasize that in \cref{sec:overview:sec:ldd-exp} we only focus on Quest 1 from the introduction, which is minimizing the loss $L$ of an LDD.

\paragraph{Reduction to Unit Lengths.}
Let us assume throughout that we only deal with unit-length graphs (i.e., where $\ell(e) = 1$ for all edges). This assumption is in fact without loss of generality, as we can simply replace each edge $e$ of length $\ell(e)$ by a path of $\ell(e)$ unit-length edges. This transformation may blow up the graph, but as we focus on the existential result first this shall not concern us.\footnote{We are assuming throughout that all edge weights are bounded by $\poly(n)$, therefore this transformation leads to a graph on at most $n_0 \leq \poly(n)$ nodes. In particular, to achieve a loss of $L = \Order(\log n \log \log n)$ on the original graph it suffices to achieve a loss of $\Order(\log n_0 \log\log n_0)$ on the transformed graph.} Whenever the LDD cuts at least one of the path edges in the transformed graph, we imagine that the entire edge in the original graph is cut. This way, when we cut each edge with probability at most $\frac{L}{D}$, in the original graph we cut each edge with probability at most $\frac{\ell(e) \cdot L}{D}$ (by a union bound) as required.

\paragraph{Reduction to Cost-Minimization.}
Perhaps it appears surprising that we claim a connection between LDDs---which are inherently \emph{probabilistic} objects---and EDs---which rather have a \emph{deterministic} flavor. As a first step of bringing these notions together consider the following lemma based on the well-studied Multiplicative Weight Update \cite{AroraHazanKale12} method.

\begin{restatable}[Multiplicative Weight Update]{lemma}{mwu} \label{lem:mwu}
Let $G = (V, E)$ be a directed graph, and suppose that for all cost functions $c : E \to [|V|^{10}]$ there is a set of cut edges $S \subseteq E$ satisfying the following two properties:
\begin{itemize}
	\item For any two nodes $u, v \in V$ that are part of the same strongly connected component in $G \setminus S$, we have $d_G(u, v) \leq D$ and $d_G(v, u) \leq D$.
	\item $c(S) \leq c(E) \cdot \frac{L}{D}$.
\end{itemize}
Then there exists a directed LDD for $G$ with loss $\Order(L)$.
\end{restatable}

Intuitively, the lemma states that in order to construct an LDD that cuts each edge with small probability $L / D$, it suffices to instead find a cut which \emph{globally} minimizes the total cost of the cut edges while ensuring that every remaining strongly connected component has diameter at most~$D$. This however comes with the price of introducing a \emph{cost function $c$} to the graph, and the goal becomes to cut edges which collectively have at most an $L / D$ fraction of the total cost. For the reader's convenience, we include a quick proof sketch of \Cref{lem:mwu}.

\begin{proof}[Proof Sketch.]
	Assign a unit cost to every edge in the graph, and consider the following iterative process. In each iteration we find a set of cut edges $S_i$ with the desired guarantees, and adjust the costs of the edges multiplicatively by doubling the cost of any every edge $e \in S_i$. Repeat this process for $R = 100\ln |E| \cdot D/L$ iterations. We claim that the uniform distribution over the cut sets~$S_i$ encountered throughout is an LDD as required. 

	Clearly the diameter condition holds for each set $S_i$, but it remains to bound the edge-cutting probability. Each iteration doubles the cost of every cut edge, and hence the total cost increases by at most $c(S_i) \leq c(E) \cdot \frac{L}{D}$, i.e., by a factor $1 + \frac{L}{D}$. After all $R$ repetitions the total cost is at most $|E| \cdot (1+L/D)^R < |E| \cdot |E|^{100} = |E|^{101}$. It follows that each edge is cut in at most $101 \log |E|$ iterations as otherwise its cost alone would already be more than $|E|^{101}$. Thus, the probability of cutting any fixed edge in a randomly sampled cut set $S_i$ is at most $101 \log |E| / R = \Order(L / D)$.
\end{proof}

Henceforth, we refer to the process of finding a cut with the properties stated in \Cref{lem:mwu} as a \emph{cost-minimizer}, which we aim to construct in the following paragraphs. That is, we now also consider a cost function $c$, and our goal is to cut edges $S$ of total cost~\smash{$c(S) \leq c(E) \cdot \frac{L}{D}$} (without worrying about a per-edge guarantee), while ensuring that the strongly connected components in~\makebox{$G \setminus S$} have diameter $\leq D$. We remark that this cost-minimization framework is quite standard.

\paragraph{Directed Expander Decomposition.}
In a leap forward, let us rename the \emph{costs $c$} to \emph{capacities~$c$}. Our idea is simple: We want to use an Expander Decomposition to remove edges of small total capacity so that all strongly connected components in the graph become expanders and thus, in particular, have small diameter.

To make this more precise, we first introduce some (standard) notation. A node set $U \subseteq V$ naturally induces a \emph{cut} (between $U$ and $\overline U = V \setminus U$). We write $c(U, \overline U)$ for the total capacity of edges crossing the cut from $U$ to $\overline U$, and define the \emph{volume} of $U$ as
\begin{equation*}
	\vol(U) = c(U, V) = \sum_{e \in E \cap (U \times V)} c(e),
\end{equation*}
and also write $\minvol(U) = \min\set{\vol(U), \vol(\overline U)}$. The \emph{sparsity} of the cut $U$ is defined by
\begin{equation*}
	\phi(U) = \frac{c(U, \overline U)}{\minvol(U)}.
\end{equation*}
We say that $U$ is \emph{$\phi$-sparse} if $\phi(U) \leq \phi$, and we call a graph without $\phi$-sparse cuts a \emph{$\phi$-expander.} The standard Expander Decomposition can be stated as follows.

\begin{lemma}[Directed Expander Decomposition] \label{lem:exp-decomp}
Let $\phi > 0$. For any directed graph $G = (V, E, c)$ there is an edge set $S \subseteq E$ of total capacity $c(S) \leq c(E) \cdot \phi \log c(E)$ such that every strongly connected component in the remaining graph $G \setminus S$ is a $\phi$-expander.
\end{lemma}

Moreover, an important property of a $\phi$-expander decomposition is that the diameter of its strongly connected components depends on $\phi$ in the following manner. 

\begin{lemma} \label{lem:exp-diam}
Any $\phi$-expander has diameter $\Order(\phi^{-1} \log \vol(V))$.
\end{lemma}

To understand the intuition behind \Cref{lem:exp-diam}, imagine that we grow a ball (i.e., a breadth-first search tree) around some node. With each step, the expansion property (related to the absence of $\phi$-sparse cuts) guarantees that we increase the explored capacity by a factor of $(1 + \phi)$; thus after $\Order(\phi \log \vol(V))$ steps we have explored essentially the entire graph.

Already at this point we have made significant progress and can recover the $\Order(\log^2 n)$-loss LDD: Apply the Expander Decomposition in \cref{lem:exp-decomp} with parameter $\phi = \log \vol(V) / D$. This removes edges with total capacity at most a $\Order(\log^2 \vol(V) / D)$-fraction of the total capacity. In the remaining graph each strongly connected component is a $\phi$-expander and thus, by \cref{lem:exp-diam}, has diameter at most $\Order(\phi^{-1} \log\vol(V)) = \Order(D)$ (by choosing the constants appropriately, this bound can be adjusted to~$\leq D$). Finally, \cref{lem:mwu} turns this into an LDD with loss $\Order(\log^2 \vol(V)) = \Order(\log^2 n)$.

Unfortunately, both log-factors in the Expander Decomposition (fraction of the cut capacity and diameter) are tight. Nevertheless, with some innovation we manage to bypass these bounds and improve the $\Order(\log^2 n)$ loss. To this end we propose a refined notion of expanders called \emph{lopsided expanders.}

\paragraph{Lopsided Expander Decomposition.}

The notion of $\phi$-sparsity defined above is oblivious to the ratio between $\vol(V)$ and $\minvol(U)$. For example, two cuts $U$ and $W$ can have the same sparsity, even though $\vol(U) \gg \vol(\overline U)$ while $\vol(W) \approx \vol(\overline W)$. It turns out that this leaves some unused potential, and that we should incentivize cutting cuts with large volume on both sides compared to more lopsided cuts with large volume only on one side.

Formally, we define \emph{$\psi$-lopsided sparsity} of a cut $U$ as 
\begin{align*}
	\psi(U) = \frac{c(U, \overline U)}{\minvol(U) \cdot \log \frac{\vol(V)}{\minvol(U)}},
\end{align*}
where we include the ratio~\smash{$\frac{\vol(V)}{\minvol(U)}$} in the denominator. Since $\vol(V)>\minvol(U)$, a cut can only have smaller lopsided sparsity than regular sparsity, so a graph with no sparse cuts may still have lopsided sparse cuts. A $\psi$-lopsided expander is defined as a graph with no $\psi$-lopsided sparse cuts, and can be thought of as a subclass of expanders which in addition to having no sparse cuts also has no cuts that are both sufficiently lopsided and sufficiently sparse.

The \emph{lopsided expander decomposition} is otherwise identical to the standard expander decomposition defined previously, except that every strongly connected component is required to be a \emph{$\psi$-lopsided} expander instead of a $\phi$-expander. Our Lopsided Expander Decomposition has the same global ``total capacity cut'' guarantee as the standard expander decomposition, as stated in the following lemma coupled with a proof sketch.

\begin{restatable}[Lopsided Expander Decomposition]{lemma}{lemLexpDecomp} \label{lem:lexp-decomp}
	Let $\psi > 0$. For any directed graph $G = (V, E, c)$ there is an edge set $S \subseteq E$ of total capacity $c(S) \leq c(E) \cdot \psi \log c(E)$ such that every strongly connected component in the remaining graph $G \setminus S$ is a $\psi$-lopsided expander.
\end{restatable}

\begin{proof}[Proof Sketch.]
	Consider the following algorithm: If there are no $\psi$-lopsided sparse cuts, then the graph is already a $\psi$-lopsided expander and we can stop. Otherwise, we cut a $\psi$-lopsided sparse cut (add the cut edges to $S$), and recurse on the remaining strongly connected components. It is clear that this eventually produces a $\psi$-lopsided expander decomposition. In order to prove that the cut edges have capacity at most $c(E) \cdot \psi \log c(E)$, we use a potential argument. We assign to each edge~$e$ a potential of $c(e) \log c(E)$. Throughout the procedure we maintain the invariant that each edge holds a potential of at least $c(e) \log \vol(C)$, where $C$ is the strongly connected component containing edge $e$. When cutting a $\psi$-lopsided cut $(U, \overline U)$ in a component $C$, an edge $e$ on the smaller side by volume (say $U$) suddenly needs to hold a potential of $c(e)\log c(U)$ instead of $c(e)\log c(C)$. In particular, the amount of excess potential we have freed is
	\begin{equation*}
		\sum_{e \in E \cap (U \times V)} c(e) \cdot (\log \vol(C) - \log \vol(U)) = \sum_{e \in E \cap (U \times V)} c(e) \cdot \log \frac{\vol(C)}{\vol(U)} = \vol(U) \cdot \log \frac{\vol(C)}{\vol(U)}.
	\end{equation*}
	Since the cut is $\psi$-lopsided sparse it follows that the total capacity of the cut edges is at most
	\begin{equation*}
		c(U, \overline U) \leq \psi \cdot \vol(U) \cdot \log \frac{\vol(C)}{\vol(U)}.
	\end{equation*}
	Thus, we can afford to donate to each cut edge $e \in S$ a potential of $c(e) / \psi$ while maintaining our invariant. Since the total potential in the graph is $c(E) \log c(E)$ and every cut edge $e \in S$ receives a potential of at least $c(e) / \psi$, we finally have that $c(E) \log c(E) \geq c(S) / \psi$ and the claim follows.
\end{proof}

The main motivation for defining $\psi$-lopsided sparsity and taking the ratio~\smash{$\frac{\vol(V)}{\minvol(U)}$} into account lies in the following lemma: Compared to standard expanders, lopsided expanders only suffer a loglog-factor in the diameter.

\begin{restatable}{lemma}{lemLexpDiam} \label{lem:lexp-diam}
	Any $\psi$-lopsided expander has diameter $\Order(\psi^{-1} \log \log \vol(V) + \log\vol(V))$.
\end{restatable}

\begin{proof}[Proof Sketch.]
	The proof is similar in spirit to the proof for standard expanders in \Cref{lem:exp-diam}, for which we restate the intuition for clarity. Imagine growing a ball (i.e., a breadth-first search tree) around some node. With each step, the expansion property (related to the absence of $\phi$-sparse cuts) guarantees that we increase the explored capacity by a factor of $(1 + \phi)$; thus after $\Order(\phi \log \vol(V))$ steps we have explored essentially the entire graph.
	
	The difficulty in adopting the same proof for $\psi$-lopsided expanders stems from the fact that $\psi$-lopsided sparsity of a cut $U \subseteq V$ depends on the volume of $U$, so the expansion of the ball around a node differs in every step, e.g., if the volume of the ball around some node is constant, after one ball-growing step it becomes $O(\psi \log \vol(V))$. We resolve this challenge by analyzing the number of steps needed for the volume of the ball to grow from $\vol(V)/2^{i+1}$ to $\vol(V)/2^{i}$. This turns out to be roughly $\lceil (\psi \cdot i)^{-1} \rceil$, implying that after
	\begin{equation*}
		\sum_{i=1}^{\log \vol(V)} \lceil (\psi \cdot i)^{-1} \rceil = \Order(\psi^{-1} \log\log \vol(V) + \log \vol(V)) 
	\end{equation*}
	steps we have explored essentially the entire graph.
\end{proof}

Putting these two lemmas for lopsided expanders together, we indeed obtain the low-loss LDD in \Cref{thm:main-existential}. Specifically, we apply the Lopsided Expander Decomposition from \cref{lem:lexp-decomp} with parameter $\psi = \log\log \vol(V) / D$. We thereby cut only a $\Order(\log \vol(V) \log\log \vol(V) / D)$-fraction of the total capacity, and end up with a graph in which every strongly connected component is a $\psi$-lopsided expander. By \cref{lem:lexp-diam} said components have diameter $\Order(\psi^{-1} \log\log \vol(V)) = \Order(D)$ (which, again, can be made $\leq D$ by adjusting constants). Plugging this procedure into \cref{lem:mwu} we conclude that there is an LDD with loss $\Order(\log \vol(V) \log\log \vol(V)) = \Order(\log n \log\log n)$, completing the proof sketch of \cref{thm:main-existential}.

\subsection{A Deterministic Algorithm} \label{sec:overviewDet}
So far we have neglected Quest~2, i.e., the design of efficient algorithms. But how far from algorithmic is this approach of \Cref{sec:overview:sec:ldd-exp} really? It turns out that implementing this framework with some simple tricks leads to the following algorithmic result.

\begin{restatable}{theorem}{thmMainDet} \label{thm:main-det}
	For every directed graph there exists a directed LDD with loss $\Order(\log n \log\log n)$ and support size $\Order(D \log n)$ that can be computed in time $\widetilde\Order(m \poly(D))$ by a \emph{deterministic} algorithm.
\end{restatable}

This theorem comes with a strength and a weakness---the strength is that the theorem is deterministic. Note that the algorithm produces the \emph{explicit} support of an LDD; in fact, the probability distribution is simply a uniform distribution over a given list of cut sets $S$. For this reason, our algorithm might lead to some derandomized applications down the road by executing an LDD-based algorithm one-by-one for all $\Order(D \log n)$ cut sets. Derandomizations of this sort are common and sought after, especially in distributed and parallel domains. The weakness is that the running time has an overhead of $\poly(D)$. We remark that almost all algorithmic applications of LDDs (with the notable exception of~\cite{BernsteinNW22}) set $D = \polylog(n)$ or $D = n^{\order(1)}$, in which case the overhead is not dramatic, and the runtime is possibly near-linear.

\paragraph{Proof Idea.}
It can be quickly verified that implementing the Multiplicative Weights Update method from \Cref{lem:mwu} is not a problem algorithmically, and requires $O(R \cdot T)$ time, where $R=O(\log |E| \cdot D/L)$ and $T$ is the time required by the cost-minimizer. Hence, only computing the Lopsided Expander Decomposition could be costly. Since this is a strictly stronger definition that the standard Expander Decomposition, at first glance it appears hopeless that we could achieve a simple algorithm that avoids the cut-matching games machinery. Luckily, we find yet another simple insight that allows us to find a simple and intuitive algorithm after all: While it may generally be hard to find a sparse cut (even NP-hard!), in our case we only ever need to find a sparse cut \emph{in a graph with diameter more than $D$.} Indeed, if the graph has already diameter at most $D$, we can stop immediately (recall that we ultimately only care about the diameter and not the expander guarantee). One way to view the algorithm is that we construct a \emph{truncated} $\psi$-lopsided expander decomposition, where we terminate prematurely if the diameter is small.

Fueled by this insight, consider the following two-phase process for computing a (truncated) $\psi$-lopsided expander decomposition:
\begin{itemize}
    \setlength\parindent{1.6em}
    \setlength\parskip{0pt}
 	\item \emph{Phase (I):} We repeatedly take an arbitrary node $v$ and grow a ball around $v$ in the hope of finding a $\psi$-lopsided sparse cut. By choosing $\psi = \Theta(\log\log \vol(V) / D)$ we can guarantee two possible outcomes: 
 	\begin{enumerate}
 		\item[(a)] We find a sparse cut after, say, $\frac{D}{4}$ steps, in which case we cut the edges crossing the cut and recur on both sides (as in the Lopsided Expander Decomposition).
 		\item[(b)] The problematic case happens: the volume of the ball reaches $\frac{3}{4} \cdot \vol(V)$. In this case we have spent linear time but made no progress in splitting off the graph, so we remember $v$ and move to Phase (II).
 	\end{enumerate}
	
 	\item \emph{Phase (II):} We compute a \emph{buffer zone} around $v$, which consists of all nodes with distance $\frac{D}{2}$ to and from $v$. We repeatedly take an arbitrary node $u$ from outside this buffer zone and grow a ball around $u$ in the hope of finding a $\psi$-lopsided sparse cut. Because we know that the $\frac{D}{4}$-radius ball around $v$ contains most of the volume of the graph, and node $u$ is picked outside of the $\frac{D}{2}$-radius buffer zone, we can prove (using similar techniques to the proof sketch of \Cref{lem:lexp-diam}) that we will always find a $\psi$-lopsided sparse cut $U$ around $u$. Upon finding it, we cut the edges crossing the cut $(U, \overline U)$ and start Phase (I) for the graph induced by nodes in $U$, while continuing Phase (II) in graph $G \setminus U$.
 	
 	When eventually there are no nodes left outside the buffer zone, we stop, since the remaining graph has diameter $\leq D$.
\end{itemize}

The correctness of the procedure above is straightforward to show: we are essentially constructing a $\psi$-lopsided expander decomposition from \Cref{lem:lexp-decomp}, but terminating prematurely if the diameter is small. Hence, the total capacity of the cut edges is at most the total capacity of the cut edges in a ``complete'' $\psi$-lopsided expander, which is a $\Order(\log \vol(V) \log\log \vol(V) / D)$-fraction of the total capacity. The diameter guarantee follows directly from the stopping condition. By plugging this procedure into the algorithmic version of \cref{lem:mwu} we arrive at an LDD with loss $\Order(\log \vol(V) \log\log \vol(V)) = \Order(\log n \log\log n)$ in time $\widetilde\Order(m \poly(D))$, completing the proof sketch of \cref{thm:main-det}.

\subsection{A Near-Optimal Randomized Algorithm} \label{sec:overviewFast}

Our third and final contribution is achieving a near-linear running time, regardless of the magnitude of the diameter, by a randomized algorithm:

\begin{restatable}{theorem}{thmMainFast} \label{thm:main-fast}
	For every directed graph there exists a directed LDD with loss $\Order(\log n \log\log n)$ which can be computed (i.e., sampled from) in expected time $\widetilde\Order(m)$.
\end{restatable}

In contrast to \cref{thm:main-existential,thm:main-det}, our approach to \cref{thm:main-fast} is rather technical. We try more directly to extend the ideas of Bernstein, Nanongkai and Wulff-Nilsen~\cite{BernsteinNW22} (which in turn borrow ideas from~\cite{BernsteinGW20}). On a very high level, their algorithm classifies nodes $v$ \emph{heavy} if it reaches more than $\frac{n}{2}$ nodes within distance $\frac{D}{2}$, or \emph{light} otherwise. We can afford to cut around light nodes $v$ (with a specifically sampled radius $r \leq D$), and recur on both sides of the cut---since $v$ is light, the inside of the cut reduces by a constant fraction which is helpful in bounding the recursion depth. And if there are only heavy nodes in the graph then we can stop as the diameter is already at most $D$---indeed, the radius-$\frac{D}{2}$ balls around heavy nodes must necessarily intersect. The radius is sampled from a geometric distribution with rate $\Order(\log n / D)$ inspired by classical LDDs in undirected graphs~\cite{Bartal96}. The two logarithmic factors stem from (i) this $\log n$ overhead in the sampling rate, and (ii) the logarithmic recursion depth.

Our approach refines these ideas by classifying not into two classes---heavy or light---but into $\log\log n$ different levels. We essentially say that a node $v$ is at level $\ell$ if it reaches roughly $n / 2^{2^\ell}$ nodes within some distance roughly $D$. We can take advantage of this in two ways: At the one extreme, we can only cut around few nodes of small level $\ell$. Indeed, when we cut around a node at level $\ell$ we remove roughly $n / 2^{2^\ell}$ nodes from the graph, so this can be repeated at most $2^{2^{\ell}} \ll n$ times. For these levels we can afford to sample the radius in the geometric distribution with a significantly smaller rate, $\Order(2^i / D)$, thereby improving~(i). At the other extreme, whenever we cut around nodes with large level $\ell$ then in the recursive call of the inside of the cut there are only~$n / 2^{2^\ell} \ll n$ nodes. This effectively reduces the recursion depth of these levels to much less than $\log n$, improving (ii). In our algorithm, we find a balance between these two extreme cases. 

Unfortunately, while this idea works out existentially, there are several issues when attempting to implement this approach in near-linear time. The first issue---initially classifying which nodes are at what level in time $\widetilde\Order(m)$---can be solved by an algorithm due to Cohen~\cite{Cohen97}. However, over the course of the algorithm when we repeatedly remove nodes and edges from the graphs the level of a node might \emph{change}. Our solution involves a cute trick: First we prove that during the execution of our algorithm the level of a node can only \emph{increase} (and never decrease). Second, instead of cutting around an arbitrary node $v$, we always pick a \emph{random} node. The argument is as follows: If the classification is still correct for at least half of the nodes, then in each step we make progress with probability $\frac12$. And otherwise we can in fact afford to reclassify by Cohen's algorithm as the level of the nodes can only increase a small number of times, leading to only few repetitions of Cohen's algorithm in total.

We omit further details here and refer to the technical \cref{sec:ldd-fast}.
\section{Preliminaries} \label{sec:prelims}
Before diving into the technical results, we state the basic graph notations used throughout the paper and recap the new non-standard definitions we have introduced throughout \Cref{sec:overview}.

\paragraph{Graphs.}
Throughout we consider directed simple graphs $G = (V, E)$, where $E \subseteq V^2$, with $n = |V|$ nodes and $m = |E|$ edges. The edges of the graph can be associated with some value: a length $\ell(e)$ or a capacity/cost $c(e)$, all of which we require to be positive. For any $U \subseteq V$, we write $\overline U = V \setminus U$. Let $G[U]$ be the subgraph induced by $U$. We denote with $\delta^{+}(U)$ the set of edges that have their starting point in $U$ and endpoint in~$\overline U$. We define $\delta^{-}(U)$ symmetrically. We also sometimes write $c(S) = \sum_{e \in S} c(e)$ (for a set of edges $S$) or $c(U, W) = \sum_{e \in E \cap (U \times W)} c(e)$ and $c(U) = c(U, U)$ (for sets of nodes $U, W$).

The distance between two nodes $v$ and $u$ is written $d_G(v,u)$ (throughout we consider only the \emph{length} functions to be relevant for distances). We may omit the subscript if it is clear from the context. The diameter of the graph is the maximum distance between any pair of nodes. For a subgraph $G'$ of $G$ we occasionally say that~$G'$ has \emph{weak diameter} $D$ if for all pairs of nodes $u, v$ in~$G'$, we have $d_G(u, v), d_G(v, u) \leq D$. A strongly connected component in a directed graph $G$ is a subgraph where for every pair of nodes $v,u$ there is a path from $v$ to $u$ and vise versa. Finally, for a radius $r \geq 0$ we write $B^+(v, r) = \set{x \in V : d_G(v, x) \leq r}$ and $B^-(v, r) = \set{y \in V : d_G(y, v) \leq r}$.

\paragraph{Polynomial Bounds.}
For graphs with edge lengths (or capacities), we assume that they are positive and the maximum edge length is bounded by $\poly(n)$. This is only for the sake of simplicity in \cref{sec:ldd-expander,sec:ldd-deterministic} (where in the more general case that all edge lengths are bounded by some threshold $W$ some logarithmic factors in $n$ become $\log (nW)$ instead), and is not necessary for our strongest LDD developed in \cref{sec:ldd-fast}.

\paragraph{Expander Graphs.}
Let $G = (V, E, \ell, c)$ be a directed graph with positive edge capacities $c$ and positive unit edge lengths $\ell$. We define the \emph{volume $\vol(U)$} by
\begin{equation*}
	\vol(U) = c(U, V) = \sum_{e \in E \cap (U \times V)} c(e),
\end{equation*}
and set $\minvol(U) = \min\set{\vol(U), \vol(\overline U)}$ where $\overline U = V \setminus U$. A node set $U$ naturally corresponds to a cut $(U, \overline U)$. The \emph{sparsity} (or \emph{conductance}) of $U$ is defined by
\begin{equation*}
	\phi(U) = \frac{c(U, \overline U)}{\minvol(U)}.
\end{equation*}
In the special cases that $U = \emptyset$ we set $\phi(U) = 1$ and in the special case that $U \neq \emptyset$ but $\vol(U) = 0$, we set $\phi(U) = 0$.
We say that $U$ is \emph{$\phi$-sparse} if $\phi(U) \leq \phi$. We say that a directed graph is a $\phi$-expander if it does not contain a $\phi$-sparse cut $U \subseteq V$. 
We define the \emph{lopsided sparsity} of $U$ as
\begin{equation*}
	\psi(U) = \frac{c(U, \overline U)}{\minvol(U) \cdot \log \frac{\vol(V)}{\minvol(U)}},
\end{equation*}
(with similar special cases), and we similarly say that $U$ is \emph{$\psi$-lopsided sparse} if $\psi(U) \leq \psi$. Finally, we call a graph a \emph{$\psi$-lopsided expander} if it does not contain a $\psi$-lopsided sparse cut $U \subseteq V$.

\section{Near-Optimal LDDs via Expander Decompositions} \label{sec:ldd-expander}

In this section, we show the existence of a near-optimal LDD, thereby proving our first main theorem: 

\thmMainExistential*

We introduce some technical lemmas in order to build up the framework for the proof of \Cref{thm:main-existential}, which can be found in the end of the section.

\subsection{Reduction to Cost-Minimizers}

\mwu*

\begin{proof}
	Let $G = (V, E)$ denote the graph for which we are supposed to design the LDD. Let us also introduce edge costs that are initially defined as $c(e) = 1$ for all edges. We will now repeatedly call the cost-minimizer, obtain a set of cut edges $S \subseteq E$, and then update the edge costs by $c(e) \gets 2 \cdot c(e)$ for all $e \in S$. We stop the process after $R = \log |E| \cdot D/L$ iterations, and let $\mathcal S$ denote the collection of all~$R$ sets~$S$ that we have obtained throughout. We claim that the uniform distribution on $\mathcal S$ is the desired LDD for $G$.
	
	It is clear that all for all sets $S \in \mathcal S$ the diameter condition is satisfied. We show that additionally for all edges $e \in E$ we have that
	\begin{equation*}
		\Pr(e \in S) \leq \frac{10L}{D},
	\end{equation*}
	for $S$ sampled uniformly from $\mathcal S$. Suppose otherwise, then in particular we have increased the cost of $e$ to at least
	\begin{equation*}
		c(e) \geq 2^{\frac{10L}{D} \cdot R} = 2^{10 \log |E|} = |E|^{10}.
	\end{equation*}
	On the other hand, let $c'$ denote the adapted costs after running the process for one iteration. Then the total cost increase is
	\begin{equation*}
		\sum_{e \in E} c'(e) - c(e) = \sum_{e \in S} c'(e) - c(e) = \sum_{e \in S} c(e) = c(S) \leq c(E) \cdot \frac{L}{D}.
	\end{equation*}
	That is, with every step of the process the total cost increases by a factor of $(1 + \frac{L}{D})$ and thus the total cost when the process stops is bounded by
	\begin{equation*}
		|E| \cdot \parens*{1 + \frac{L}{D}}^R \leq |E| \cdot e^{\frac{L}{D} \cdot R} = |E| \cdot e^{\log |E|} \leq |E|^3,
	\end{equation*}
	leading to a contradiction. The same argument shows that all costs are bounded by $|E|^3 \leq |V|^6$ throughout.
\end{proof}

\lemLexpDecomp*

\begin{proof}
Consider the following algorithm: If there is no $\psi$-lopsided sparse cut then the graph is a $\psi$-lopsided expander by definition and we stop. Otherwise, there exists a $\psi$-lopsided sparse cut~$(U, \overline U)$. We then distinguish two cases: If $\vol(U) \leq \vol(\overline U)$ then we remove all edges from~$U$ to~$\overline U$, and otherwise we remove all edges from~$\overline U$ to~$U$ (in both cases placing these edges in~$S$). Then we recursively continue on all strongly connected components in the remaining graph~$G \setminus S$. 

It is clear that all strongly connected components in the remaining graph $G \setminus S$ are $\psi$-lopsided expanders, but it remains to show that we cut edges with total capacity at most $c(E) \cdot \psi \log c(E)$. Imagine that initially we associate to each edge $e$ a \emph{potential} of~\makebox{$c(e) \cdot \log c(E)$}. The total initial potential is thus $\sum_e c(e) \log c(E) = c(E) \log c(E)$. Throughout the procedure we maintain the invariant that each edge holds a potential of at least $c(e) \log \vol(C)$, where $C$ is the strongly connected component containing edge $e$. Focus on any recursion step and its current strongly connected component~$C$, and let $C = U \sqcup \overline U$ denote the current $\psi$-lopsided sparse cut. Assume first that $\vol(U) \leq \vol(\overline U)$. Observe that an edge $e \in U$ suddenly needs to hold a potential of $c(e)\log c(U)$ instead of $c(e)\log c(C)$. Hence, the amount of freed potential in $U$ is at least
\begin{align*}
	\sum_{e \in E \cap (U \times V)} c(e) (\log \vol(C) - \log \vol(U)) &=
	\sum_{e \in E \cap (U \times V)} c(e)  \cdot \log \frac{\vol(C)}{\vol(U)} = \vol(U) \cdot \log \frac{\vol(C)}{\vol(U)}.
\end{align*}
On the other hand, since $(U, \overline U)$ is a $\psi$-lopsided sparse cut we have that
\begin{equation*}
	\psi \geq \psi(U) = \frac{c(U, \overline U)}{\minvol(U) \cdot \log \frac{\vol(C)}{\minvol(U)}} = \frac{c(U, \overline U)}{\vol(U) \cdot \log \frac{\vol(C)}{\vol(U)}}.
\end{equation*}
Putting these together, this means any cut edge $e$ from $U$ to $\overline U$ can get ``paid'' a potential of~\smash{$c(e) \cdot \psi^{-1}$} while still maintaining the potential invariant. (Note that here we only exploit the potential freed by the smaller side of the cut $U$, and forget about the overshoot potential in the larger side $\overline U$.) A symmetric argument applies when $\vol(U) < \vol(\overline U)$.

All in all, we start with a total potential of $c(E) \log c(E)$ and pay for each cut edge $e \in S$ with a potential of at least~\smash{$c(e) \cdot \psi^{-1}$}. This implies that $c(E) \log c(E) \geq c(S) \cdot \psi^{-1}$ and the claim follows.
\end{proof}

To prove that lopsided expanders have small diameter, we first establish the following technical lemma. 

\begin{lemma} \label{lem:lopsided-expansion}
Let $G = (V, E, c)$ be a directed graph and let $\psi > 0$. For any node $v \in V$ there is some radius $R = \Order(\psi^{-1} \log\log \vol(V) + \log \vol(V))$ such that one of the following two properties holds:
\begin{itemize}
	\item $\vol(B^+(v, R)) \geq \frac{1}{2} \cdot \vol(V)$, or
	\item $\psi(B^+(v, r)) \leq \psi$ for some $0 \leq r \leq R$.
\end{itemize}
\end{lemma}
\begin{proof}
We write $\Delta_i = \ceil{\frac{1}{i \psi}}$ and define the radii~\smash{$1 = r_{\ceil{\log \vol(V)}} \leq \dots \leq r_1$} by $r_i = r_{i+1} + \Delta_i$. We prove by induction that~\smash{$\vol(B^+(v, r_i)) \geq 2^{-i} \cdot \vol(V)$}, or alternatively that we find a sparse cut. This is clearly true in the base case for $i = \ceil{\log \vol(V)}$: Either $\vol(B^+(v, 1)) \geq 1$ or $v$ is an isolated node and therefore $\psi(B^+(v, r)) = 0$.

For the inductive case, suppose for the sake of contradiction that $\vol(B^+(v, r_i)) < 2^{-i} \cdot \vol(V)$. By induction we know however that $\vol(B^+(v, r_{i+1})) \geq 2^{-i-1} \cdot \vol(V)$. It follows there is some radius $r_{i+1} \leq r < r_i = r_{i+1} + \Delta_i$ such that
\begin{equation*}
	\frac{\vol(B^+(v, r + 1))}{\vol(B^+(v, r))} \leq 2^{1/\Delta_i} \leq 1 + \frac{1}{\Delta_i}.
\end{equation*}
It follows that
\begin{equation*}
	c(B^+(v, r), \overline{B^+(v, r)}) = \vol(B^+(v, r + 1)) - \vol(B^+(v, r)) \leq \frac{\vol(B^+(v, r))}{\Delta_i}.
\end{equation*}
Therefore the cut induced by $B^+(v, r)$ has lopsided sparsity
\begin{align*}
	\psi(B^+(v, r)) &= \frac{c(B^+(v, r), \overline{B^+(v, r)})}{\minvol(B^+(v, r)) \log \frac{\vol(V)}{\minvol(B^+(v, r))}} \\
	&\leq \frac{c(B^+(v, r), \overline{B^+(v, r)})}{\vol(B^+(v, r)) \log \frac{\vol(V)}{\vol(B^+(v, r))}} \\
	&\leq \frac{1}{\Delta_i \cdot \log \frac{\vol(V)}{\vol(B^+(v, r))}} \\
	&\leq \frac{1}{\Delta_i \cdot \log \frac{\vol(V)}{\vol(B^+(v, r_i))}} \\
	&\leq \frac{1}{\Delta_i \cdot \log(2^i)} \\
	&\leq \psi.
\end{align*}
Here, in the second step we have used that $\minvol(B^+(v, r)) = \vol(B^+(v, r))$ as in the opposite case we have $\vol(B^+(v, r)) \geq \frac{1}{2} \cdot \vol(V)$ which also proves the claim. This finally leads to a contradiction since we assume that the graph is a $\psi$-lopsided expander and thus does not contain $\psi$-lopsided sparse cuts.

In summary, the induction shows that $\vol(B^+(v, r_1)) \geq \frac{1}{2} \cdot \vol(V)$ and thus we may choose $R = r_1$. To prove that $R$ is as claimed, consider the following calculation:
\begin{align*}
	R &= 2 + \sum_{i=1}^{\ceil{\log \vol(V)}} \Delta_i \\
	&= 2 + \sum_{i=1}^{\ceil{\log \vol(V)}} \ceil*{\frac{1}{i \psi}} \\
	&\leq 2 + \ceil{\log \vol(V)} + \sum_{i=1}^{\ceil{\log \vol(V)}} \frac{1}{i \psi} \\
	&\leq \Order(\log\vol(V) + \psi^{-1} \log\log\vol(V)),
\end{align*}
using the well-known fact that the harmonic numbers are bounded by $\sum_{k=1}^n 1/k = \Order(\log n)$.
\end{proof}

One can easily strengthen the lemma as follows. This insight will play a role in the next \cref{sec:ldd-deterministic} in the construction of the deterministic algorithm.

\begin{lemma} \label{lem:lopsided-expansion-boosted}
Let $G = (V, E, c)$ be a directed graph and let $\psi > 0$ and $0 < \alpha < 1$. For any node $v \in V$ there is some radius $R = \Order(\psi^{-1} \log\log \vol(V) + \psi^{-1} \alpha^{-1} + \log \vol(V))$ such that one of the following two properties holds:
\begin{itemize}
	\item $\vol(B^+(v, R)) \geq (1 - \alpha) \cdot \vol(V)$, or
	\item $\psi(B^+(v, r)) \leq \psi$ for some $0 \leq r \leq R$.
\end{itemize}
\end{lemma}
\begin{proof}
Applying the previous lemma with parameter $\psi$ yields $R' = \Order(\psi^{-1} \log\log \vol(V) + \log \vol(V))$ such that either $\vol(B^+(v, R')) \geq \frac{1}{2} \cdot \vol(V)$, or $\psi(B^+(v, r)) \leq \psi$ for some $0 \leq r \leq R'$. In the latter case we are immediately done, so suppose that we are in the former case.

Let $\Delta = \ceil{2 \alpha^{-1} \psi^{-1}}$ and let $R = R' + \Delta$. If $\vol(B^+(v, R)) \geq (1 - \alpha) \cdot \vol(V)$ then we have shown the first case and are done. So suppose that otherwise $\vol(B^+(v, R)) \leq (1 - \alpha) \cdot \vol(V)$. Then due to the trivial bound $\vol(B^+(v, R)) \leq \vol(V)$, there is some radius $R' \leq r \leq R = R' + \Delta$ with
\begin{equation*}
	\frac{\vol(B^+(v, r + 1))}{\vol(B^+(v, r))} \leq 2^{1/\Delta} \leq 1 + \frac{1}{\Delta},
\end{equation*}
and hence,
\begin{equation*}
	c(B^+(v, r), \overline{B^+(v, r)}) = \vol(B^+(v, r + 1)) - \vol(B^+(v, r)) \leq \frac{\vol(B^+(v, r))}{\Delta}.
\end{equation*}
For this radius $r$ it further holds that $\vol(B^+(v, r)) \leq (1 - \alpha) \cdot \vol(V)$ and thus $\vol(\overline{B^+(v, r)}) \geq \alpha \cdot \vol(V)$. In particular, we have that $\minvol(B^+(v, r)) \geq \frac{\alpha}{2} \cdot \vol(V)$. Putting these statements together, we have that
\begin{align*}
	\psi(B^+(v, r))
	&= \frac{c(B^+(v, r), \overline{B^+(v, r)})}{\minvol(B^+(v, r)) \log \frac{\vol(V)}{\minvol(B^+(v, r))}} \\
	&\leq \frac{\vol(B^+(v, r))}{\Delta \cdot \frac{\alpha}{2} \cdot \vol(B^+(v, r)) \log \frac{\vol(V)}{\minvol(B^+(v, r))}} \\
	&\leq \frac{1}{\Delta \cdot \frac{\alpha}{2}} \\
	&\leq \psi,
\end{align*}
witnessing indeed the desired sparse lopsided cut.
\end{proof}

\lemLexpDiam*

\begin{proof}
	Take an arbitrary pair of nodes $v, u$. Applying \cref{lem:lopsided-expansion-boosted} with parameters $\psi$ and $\alpha = \frac{1}{4}$, say, yields a radius $R = \Order(\psi^{-1} \log\log \vol(V) + \log \vol(V))$ such that
	\begin{equation*}
		\vol(B^+(v, R)) \geq \frac{3}{4} \cdot \vol(V),
	\end{equation*}
	and symmetrically,
	\begin{equation*}
		\vol(B^-(u, R)) \geq \frac{3}{4} \cdot \vol(V).
	\end{equation*}
	Therefore, there is some edge $e = (x, y)$ contributing to both of these volumes. Thus $x \in B^+(v, R)$ and $y \in B^-(u, R)$. It follows that
	\begin{equation*}
		d_G(u, v) \leq d_G(u, x) + d_G(x, y) +  d_G(y, u) \leq R + 1 + R = \Order(\psi^{-1} \log\log \vol(V) + \log \vol(V)).
	\end{equation*}
	Since the nodes $u, v$ were chosen arbitrarily this establishes the claimed diameter bound.
\end{proof}

\begin{proof}[Proof of \cref{thm:main-existential}]
Let $G = (V, E, \ell)$ be a directed graph with positive edge lengths. We show that there is an LDD with loss $\Order(\log n \log\log n)$ for $G$. We first deal with two trivial cases: First, if $D \leq \log n / \gamma$ (for some constant $\gamma > 0$ to be determined later) then we simply remove all edges and stop. Second, we remove all edges with length more than $D$ from the graph. In both cases edges can be deleted with probability $1$ without harm.

Next, we transform the graph into $G'$ by replacing each $e$ by a path of $\ell(e)$ unit-length edges. In the following it suffices to design an LDD for the augmented graph; if that LDD cuts any of the edges along the path corresponding to an original edge $e$ we will cut $e$ entirely. An LDD with loss~$L$ in the augmented graph will thus delete an original edge with probability at most $\ell(e) \cdot \frac{L}{D}$ by a union bound. All in all, this transformation blows up the number of nodes and edges in the graph by a factor of at most $D$ (since we removed edges with larger length). Recall that we throughout assume that $D \leq n^c$, for some constant $c$, and thus $|V'| \leq n^{\Order(1)}$.

By \cref{lem:mwu} we further reduce the existence of an LDD of $G'$ to the following cost-minimizer task: View $G'$ as an edge-capacitated graph $G' = (V', E', c)$ for some capacities~\makebox{$c : E' \to [|V'|^{10}]$}. In particular, under this capacity function $G'$ has volume $\vol(V') \leq |V'|^2 \cdot |V'|^{10} = |V'|^{12} = n^{\Order(1)}$. The goal is to delete edges $S \subseteq E'$ in $G'$ so all remaining strongly connected components have (weak) diameter at most $D$, and the total cost of all deleted edges is only $c(S) \leq c(E) \cdot \frac{L}{D}$.

Finally, we apply the Lopsided Expander Decomposition from \cref{lem:lexp-decomp} on $G'$. Specifically, we define
\begin{equation*}
	\psi = \frac{\log\log \vol(V')}{\epsilon D}
\end{equation*}
for some constant $\epsilon > 0$ to be determined later. The Expander Decomposition then cuts edges $S \subseteq E'$ so that each remaining strongly connected component is $\psi$-lopsided expander. Thus, by \cref{lem:lexp-diam} each strongly connected component has diameter
\begin{equation*}
	\Order(\psi^{-1} \log\log \vol(V') + \log \vol(V')) = \Order(\epsilon D + \gamma D).
\end{equation*}
By choosing the constants $\epsilon$ and $\gamma$ to be sufficiently small, the diameter bound becomes $D$ as desired. Moreover, \cref{lem:exp-decomp} guarantees that we cut edges of total capacity
\begin{equation*}
	c(S) \leq c(E') \cdot \psi \log \vol(V') \leq c(E') \cdot \frac{\log \vol(V') \log\log \vol(V')}{\epsilon D},
\end{equation*}
which becomes $\frac{L}{D}$ by choosing
\begin{equation*}
	L = \frac{\log \vol(V') \log\log \vol(V')}{\epsilon} \leq \frac{\log |V'|^{12} \log\log |V'|^{12}}{\epsilon} = \Order(\log n \log\log n)
\end{equation*}
as planned.
\end{proof}
\section{Near-Optimal LDDs Deterministically} \label{sec:ldd-deterministic}
In this section, we present the deterministic algorithm for computing a near-optimal LDD, thereby proving our second main theorem:

\thmMainDet*

To this end, we utilize many of the same building blocks we have already introduced in \Cref{sec:ldd-expander}. In particular, we follow the general framework of Multiplicative Weights Update to reduce the computation of an LDD to solving the cost-minimizing task. The full proof of \Cref{thm:main-det} can be found in the end of the section. The following lemma restates the MWU method algorithmically; we omit a proof as it follows exactly the proof of \cref{lem:mwu} in \Cref{sec:ldd-expander}. 

\begin{lemma}[Algorithmic Multiplicative Weight Update]
Let $G = (V, E)$ be a directed graph and let $D \geq 1$. Suppose that there is an algorithm $\mathcal A$ that, given $G$, $D$ and a cost function $c : E \to [|V|^{10}]$, computes a set of edges $S \subseteq E$ satisfying the following properties:
\begin{itemize}
	\item For any two nodes $u, v \in V$ that are part of the same strongly connected component in $G \setminus S$, we have $d_G(u, v) \leq D$ and $d_G(v, u) \leq D$.
	\item $c(S) \leq c(E) \cdot \frac{L}{D}$.
\end{itemize}
Then there is a deterministic algorithm to compute an LDD with loss $\Order(L)$ for $G$ (i.e., we compute the full support of a uniform distribution over $\Order(D \log n)$ cut sets). It runs in time~\smash{$\widetilde\Order(m D)$} and issues $\Order(D \log n)$ oracle calls to $\mathcal A$.
\end{lemma}

For the remainder of this section we will therefore focus on the same cost minimizer setting: Given a directed graph $G = (V, E, c)$ with edge capacities (and unit lengths), the goal is select a set of cut edges $S \subseteq E$ such that $c(S) \leq c(E) \cdot \Order(\frac{1}{D} \cdot \log n \log\log n)$ and such that all strongly connected components in the remaining graph $G \setminus S$ have (weak) diameter at most $D$.

The following lemma is a consequence of the lopsided expander machinery set up before:

\begin{lemma}[Finding Sparse Cuts] \label{lem:sparse-cut-det}
Let $G = (V, E, c)$ be a directed graph, let $D \geq \log \vol(V)$ and let $v \in V$. Then there is some $\psi = \Order(\frac{1}{D} \cdot \log\log \vol(V))$ and an algorithm to determine which of the following cases applies:
\begin{enumerate}[label=(\roman*)]
	\item There is a radius $0 \leq r \leq D$ with $\psi(B^+(v, r)) \leq \psi$ and $c(B^+(v, r)) \leq 0.95 \cdot \vol(V)$.\\(In this case the algorithm runs in linear time in the number of edges incident to $B^+(v, r)$.)
	\item Or, there is a radius $0 \leq r \leq D$ such that $\psi(\overline{B^-(v, r)}) \leq \psi$ and $c(B^-(v, r)) \leq 0.95 \cdot \vol(V)$.\\(In this case the algorithm runs in linear time in the number of edges incident to $B^-(v, r)$.)
	\item Or, $c(B^+(v, D) \cap B^-(v, D)) \geq 0.9 \cdot \vol(V)$.\\(In this case the algorithm runs in time~\smash{$\Order(m)$}.)
\end{enumerate}
\end{lemma}
\begin{proof}
Existentially the statement follows from \cref{lem:lopsided-expansion-boosted} applied with parameters $\psi = \Theta(\frac{1}{D} \cdot \log\log \vol(V))$ (with an appropriately large hidden constant) and $\alpha = 0.05$. This lemma states that $\psi(B^+(v, r)) \leq \psi$ for some radius $0 \leq r \leq D - 1$ (proving (i)) or that~\smash{$\vol(B^+(v, D - 1)) \geq 0.95 \vol(V)$}. Applying the same statement to the reverse graph similarly gives that $\psi(\overline{B^-(v, r)}) \leq \psi$ for some radius $0 \leq r \leq D - 1$ (proving (ii)) or that~\smash{$c(\overline{B^-(v, D - 1)}, B^-(v, D - 1)) \geq 0.95 \cdot \vol(V)$}. In the only remaining case we thus have both
\begin{align*}
	c(B^+(v, D - 1), V) \geq 0.95 \cdot \vol(V)\quad\text{and}\quad c(\overline{B^-(v, D - 1)}, V) \geq 0.95 \cdot \vol(V).
\end{align*}
Combining both statements we obtain that edges of total capacity at least $0.9 \cdot \vol(V)$ must lie in $B^+(v, D - 1) \times B^-(v, D - 1)$. In particular, it follows that $c(B^+(v, D) \cap B^-(v, D)) \geq 0.9 \cdot \vol(V)$ thereby proving (iii).

To make the lemma algorithmic, we simultaneously grow an out-ball $B^+(v, r^+)$ and an in-ball~$B^-(v, r^-)$ around the node $v$. Explicitly, we start with $r^+ = 0$ and increase $r^+$ step by step to compute $B^+(v, r^+)$ (with breadth-first search). We can, without overhead, keep track of the current volume of $B^+(v, r^+)$. If we at some point encounter that $B^+(v, r)$ is $\psi$-lopsided sparse, then we stop and report output (i). Similarly, we start with $r^- = 0$ and step by step explore $B^-(v, r^-)$. If at some point $B^-(v, r^-)$ becomes a $\psi$-lopsided sparse cut, we stop and report answer (ii). We interleave these two computations so that when we output (i) the overhead of exploring the in-ball~$B^-(v, r^-)$ incurs only a constant factor in the running time, and similarly for (ii). In the remaining case where we have not encountered a sparse cut in the graph before reaching $r^+ = r^- = D$, we report (iii). In this case we indeed spend time at most $\Order(m)$.
\end{proof}

Having established \cref{lem:sparse-cut-det}, now consider the algorithm in \cref{alg:det}. In summary, it runs in two phases. In Phase (I) we first repeatedly select a node $v$ and attempt to cut a lopsided sparse cut around $v$ (i.e., we cut the edges in $\delta^+(B^+(v, r))$ for some radius $r$). We only execute these cuts, however, until we find a node $z$ for which $c(B^+(z, D') \cap B^-(z, D')) \geq 0.9 \cdot \vol(V)$---that is, both the radius-$D'$ out- and in-balls of $z$ make up for a big constant fraction of the entire graph. We call $z$ a \emph{center} node and move on to phase (II). In this phase we repeat the same steps as in Phase (I), but we only choose nodes $v$ that have distance at least $2D'$ (in one direction or the other) to the center $z$. The intuition is that we can never find a second node $z'$ which equally makes up for the entire graph, as then $z$ and $z'$ would have to be connected by a short path. In the remainder of this section we formally analyze \cref{alg:det}.

\begin{algorithm}[t]
	\caption{The deterministic near-optimal LDD, see \Cref{thm:main-det}.} \label{alg:det}
	\begin{enumerate}[label=\arabic*.]
		\item[(I)] Repeat the following steps: Take an arbitrary node $v \in V$ and apply \cref{lem:sparse-cut-det} with parameter \smash{$D' = \floor{\frac{D}{4}}$}. Depending on the output execute the following steps:
		\begin{enumerate}[label=(\roman*)]
			\item Cut all edges in $\delta^+(B^+(v, r))$, recurse on the induced graph $G[B^+(v, r)]$, then remove all nodes in $B^+(v, r)$ from the graph.
			\item Cut all edges $\delta^-(B^-(v, r))$, recurse on the induced graph $G[B^-(v, r)]$, then remove all nodes in $B^-(v, r)$ from the graph.
			\item Remember $z \gets v$ (called the \emph{center} node) and continue with Phase (II).
		\end{enumerate}
		\item[(II)] Compute the sets
		\begin{align*}
			X &= B^+(z, D') \cap B^-(z, D'), \\
			Y &= B^+(z, 2D') \cap B^-(z, 2D').
		\end{align*}
		Then repeat the following steps while there still exists nodes in $V \setminus Y$: Take an arbitrary node~\makebox{$u \in V \setminus Y$} and apply \cref{lem:sparse-cut-det} with parameter~$D'$. Depending on the output execute the following steps:
		\begin{enumerate}[label=(\roman*)]
			\item Cut all edges in $\delta^+(B^+(v, r))$, recurse on the induced graph $G[B^+(v, r)]$, then remove all nodes in $B^+(v, r)$ from the graph.
			\item Cut all edges $\delta^-(B^-(v, r))$, recurse on the induced graph $G[B^-(v, r)]$, then remove all nodes in $B^-(v, r)$ from the graph.
		\end{enumerate}
	\end{enumerate}
\end{algorithm}

\begin{lemma}[Total Cost of \cref{alg:det}] \label{lem:ldd-det-cost}
Let $S \subseteq E$ denote the set of edges cut by \cref{alg:det}. Then $c(S) \leq c(E) \cdot \Order(\frac{1}{D} \cdot \log \vol(V) \log\log \vol(V))$.
\end{lemma}
\begin{proof}
Observe that the algorithm only cuts the edges of $\psi$-lopsided sparse cuts for some parameter $\psi = \Order(\frac{1}{D} \cdot \log\log \vol(V))$. Indeed, the algorithm only cuts edges from $B^+(v, r)$ to~\smash{$\overline{B^+(v, r)}$} in subcase~(i), and edges from~\smash{$\overline{B^-(v, r)}$} to $B^-(v, r)$ in subcase~(ii). In both these cases \cref{lem:sparse-cut-det} gives exactly the guarantee that these respective cuts are $\psi$-lopsided sparse.

With this in mind, we can apply exactly the same potential argument as in the proof of \cref{lem:exp-decomp}. To avoid repetitions, we only give a quick reminder here: We initially associate to each edge a potential of~\smash{$c(e) \cdot \log\vol(V)$}. Then, following the calculations as in \cref{lem:exp-decomp}, we can free a potential of at least $c(e) / \psi$ for each cut edge, proving that $c(S) \leq c(E) \cdot \psi \log\vol(V)$ as claimed.
\end{proof}

\begin{lemma}[Well-Definedness of \cref{alg:det}] \label{lem:ldd-det-well-defined}
While executing Phase~(II) of \cref{alg:det}, the subcase~(iii) never happens.
\end{lemma}
\begin{proof}
Let $G = (V, E)$ denote the graph at the transition from Phase (I) to Phase (II). Suppose for contradiction that during the execution of Phase (II) we find a node which falls into case~(iii). That is, let $G' = (V', E')$ denote the graph remaining at this point in the execution of the algorithm, and suppose that there is a node $v \in V'$ for which $\vol(G'[B^+(v, D') \cap B^-(v, D')]) \geq 0.9 \cdot \vol(G')$. Since we have picked~\makebox{$v \not\in Y$}, we either have that $d_G(z, v) > 2D'$ or that $d_G(v, z) > 2D'$; focus on the former case. Then for all~\makebox{$x \in X$}, we have $d_G(x, v) > D'$ (as otherwise~\makebox{$d_G(z, v) \leq d_G(z, x) + d_G(x, v) \leq 2D'$}), and thus $X$ and $B^-(v, D')$ are disjoint. But this leads to a contradiction as supposedly both $\vol(G[X]) \geq 0.9 \vol(G)$ and $\vol(G[B^-(v, D')]) \geq 0.9 \vol(G') \geq 0.9 \cdot 0.9 \cdot \vol(G) \geq 0.8 \vol(G)$, leading to a total volume of more than $\vol(G)$. Here in the last step we have used that invariantly $\vol(G') \geq 0.9 \cdot \vol(G)$ since the algorithm can never cut nodes from $X$. The remaining case is symmetric.
\end{proof}

\begin{lemma}[Correctness of \cref{alg:det}] \label{lem:ldd-det-correctness}
Let $S \subseteq E$ denote the set of edges cut by \cref{alg:det}. Then for any two nodes $u, w$ in the same strongly connected component in $G \setminus S$, it holds that $d_G(u, w) \leq D$ and $d_G(w, u) \leq D$.
\end{lemma}
\begin{proof}
We prove the statement by induction. The base case is clear for an appropriate implementation of constant-size graphs. For the inductive step, consider an execution of \cref{alg:det} and an arbitrary pair of nodes $u, w$. Whenever the algorithm cuts a ball $B^+(v, r)$ and we have $u \in B^+(v, r)$ and $w \not\in B^+(v, r)$, then $u$ and $w$ do not end up in the same strongly connected component in~$G \setminus S$. If instead both $u, w \in B^+(v, r)$ then the claim follows by induction as the algorithm recurses on the subgraph induced by $B^+(v, r)$. The only remaining case is if both $u$ and $w$ are never cut during the execution of the algorithm. Clearly the algorithm cannot have terminated after Phase (I) (as then there would be no nodes left in the graph), so the algorithm has reached Phase (II). Moreover, we have $u, w \in Y$ as otherwise the algorithm would not have terminated yet. But then by definition, we have $d_G(v, u) \leq 2D'$, $d_G(u, v) \leq 2D'$, $d_G(w, v) \leq 2D'$ and $d_G(v, w) \leq 2D'$. Putting all these together we have that $d_G(u, w) \leq 4D' \leq D$ and $d_G(w, u) \leq 4D' \leq D$ as claimed.
\end{proof}

\begin{lemma}[Running Time of \cref{alg:det}] \label{lem:ldd-det-time}
\cref{alg:det} runs in time $\Order(m \log \vol(V))$.
\end{lemma}
\begin{proof}
Consider one execution of \cref{alg:det}. We spend time $\Order(m)$ once to compute the sets $X$ and $Y$ at the beginning of Phase (II). Other than that, all steps only run in local fragments of the graph. Specifically, whenever the algorithm cuts a ball $B^+(v, r)$ we spend time $\Order(|\delta^+(B^+(v, r))|)$ by \cref{lem:sparse-cut-det} (i.e., time proportional to the number of edges incident to $B^+(v, r)$), but then we delete all nodes in $B^+(v, r)$ (and thereby also all edges in $\delta^+(B^+(v, r))$). We can thus express the running time $T(m)$ by the recurrence
\begin{equation*}
	T(m, C) \leq \Order(m) + \sum_i T(m_i, C_i),
\end{equation*}
where $m_i$ is the number of edges and $C_i$ is the total capacity of the $i$-th recursive call. Clearly we have that $\sum_i m_i \leq m$. Moreover, we only cut balls with $c(B^+(v, r)) \leq 0.95 \cdot \vol(V)$ by \cref{lem:sparse-cut-det} (and similarly for the in-balls $B^-(v, r)$). This drop in capacity bounds the recursion depth by $\Order(\log \vol(V))$ and so the recursion solves to $\Order(m \log |V|)$.
\end{proof}

This completes the analysis of \cref{alg:det} and puts us in the position of completing the proof of \cref{thm:main-det}.

\begin{proof}[Proof of \cref{thm:main-det}]
As before, we first turn the given graph into a unit-length graph which only blows up the number of nodes and edges by a factor of $D$ (all edges with length $> D$ can anyways be removed for free). Then, by \cref{lem:mwu} it suffices to design a deterministic algorithm for the cost minimization problem (indeed, this algorithm is then turned into a deterministic LDD algorithm by \cref{lem:mwu}, at the cost of another factor-$\Order(D \log n)$ blow-up). Finally, we can assume that~\makebox{$D \geq \log \vol(V)$} in the remaining task, as otherwise it is within our budget to simply remove all edges.

To solve the cost minimization problem, we run \cref{alg:det}. Let $S \subseteq E$ denote the edges cut by \cref{alg:det}. Then, by \cref{lem:ldd-det-correctness} indeed the remaining graph $G \setminus S$ has (weak) diameter at most~$D$. By \cref{lem:ldd-det-cost} the total capacity of the cut edges is $c(S) \leq c(E) \cdot \Order(\frac{1}{D} \cdot \log\vol(V) \log\log\vol(V))$. And by \cref{lem:ldd-det-time} the running time is near-linear in the number of edges. All three properties are as required by \cref{lem:mwu}, ultimately leading to an LDD algorithm with loss $\Order(\log n \log\log n)$ (using that $\log \vol(V) = \Order(\log n)$) and in time $\widetilde\Order(m D^2)$.
\end{proof}
\section{Near-Optimal LDDs in Near-Linear Time} \label{sec:ldd-fast}
In this section we show that near-optimal LDDs can even be computed in near-linear (i.e., near-optimal) running time. Formally, the goal of this section is to prove the last of our main theorems:

\thmMainFast*

We structure this section as follows: In \cref{sec:ldd-fast:sec:cutting} we first prove an important cutting lemma (\cref{lem:cut-light}), which we then apply in \cref{sec:ldd-fast:sec:algo} to prove \cref{thm:main-fast}.

\subsection{Cutting Lemma} \label{sec:ldd-fast:sec:cutting}
We rely on the following result due to Cohen that we can approximate efficiently the sizes of out-balls $\Bout(v, r)$ simultaneously for all nodes $v$. Note that we can equally approximate the size of the in-balls $\Bin(v, r)$ by applying \cref{lem:cohen} on the reverse graph.

\begin{lemma}[Cohen's Algorithm~{{\cite[Theorem~5.1]{Cohen97}}}] \label{lem:cohen}
Let $G$ be a directed weighted graph, let $r \geq 0$ and~\makebox{$\epsilon > 0$}. There is an algorithm that runs in time $\Order(m \epsilon^{-2} \log^3 n)$ and computes approximations~$b(v)$ satisfying with high probability that
\begin{equation*}
    (1 - \epsilon) |\Bout(v, r)| \leq b(v) \leq (1 + \epsilon) |\Bout(v, r)|.
\end{equation*}
\end{lemma}

The goal of this subsection is to establish the following lemma:

\begin{lemma}[Cutting Light Nodes] \label{lem:cut-light}
Let $G = (V, E)$ be a directed weighted graph, and consider the parameters $\delta > 0$, $0 \leq r_0 < r_1$, $1 \leq s_1 < s_0 \leq m$. There is an algorithm that computes a set of cut edges $S \subseteq E$ and sets of remaining vertices $R \subseteq V$ such that:
\begin{enumerate}[label=(\roman*)]
    \item Each strongly connected component in $G \setminus S$ either only contains nodes from $R$ or only nodes from $V \setminus R$. In the latter case it contains at most $\frac32 \cdot \frac{m}{s_1}$ edges.
    \item For every node $v \in R$ with $|\Bout_G(v, r_1)| \leq \frac{m}{s_1}$, it holds that $|\Bout_{G[R]}(v, r_0)| \leq \frac{m}{s_0}$.
    \item None of the edges in $S$ has its source in $R$ (i.e., $S \subseteq (V \setminus R) \times V$). Moreover, for each edge $e = (x, y) \in E$ it holds that:
    \begin{equation*}
        \Pr(e \in S \mid x \not\in R) \leq \frac{w_e \ln(2s_0 / \delta)}{r_1 - r_0}.
    \end{equation*}
\end{enumerate}
With probability at most $\delta$ the algorithm instead returns ``fail'', and with probability at most~$\frac{1}{\poly(n)}$ the algorithm returns an arbitrary answer. The running time is~\smash{$\Order(m \log^4 n)$}.
\end{lemma}

\begin{proof}
Consider the pseudocode in \cref{alg:cut-light}. Throughout we maintain a set of cut edges~\makebox{$S \subseteq E$} that is initially empty. We classify nodes as \emph{good} or \emph{bad} depending on whether their out-balls~$\Bout(v, r_0)$ and~$\Bout(v, r_1)$ satisfy certain size requirements. Then, in several iterations the algorithm samples a good node $v$ and attempts to \emph{cut} around $v$ in the sense that we sample a radius~\smash{$r_0 \leq r \leq r_1$}, include the boundary edges $\delta^+(\Bout(v, r))$ into $S$, and then remove $\Bout(v, r)$ from the graph.

\begin{algorithm}[t]
\caption{Implements the cutting procedure from \cref{lem:cut-light}.} \label{alg:cut-light}
\begin{enumerate}[label=\arabic*.]
    \item Run Cohen's algorithm (\cref{lem:cohen} with parameter $\epsilon = \frac18$) on $G$ to compute approximations~$b_1(v)$ satisfying that $\tfrac78 \cdot |\Bout(v, r_1)| \leq b_1(v) \leq \tfrac98 \cdot |\Bout(v, r_1)|$. Mark all nodes $v$ satisfying that $b_1(v) \leq \frac98 \cdot \frac{m}{s_1}$ as \emph{good} and all other nodes as \emph{bad.}
    \item Repeat the following steps $\log n + 1$ times:
    \begin{enumerate}[label=2.\arabic*.]
        \item Repeat the following steps $s_0 \cdot 100 \log n$ times:
        \begin{enumerate}[label=2.1.\arabic*.]
            \item Among the nodes that are marked good, pick a uniformly random node $v$.
            \item Test whether $|\Bout(v, r_0)| \geq \frac{1}{2} \cdot \frac{m}{s_0}$, and otherwise continue with the next iteration of step (b).
            \item Sample $r \sim r_0 + \geom(p)$ (where~\smash{$p = \frac{\ln(2s_0 / \delta)}{r_1 - r_0}$}) and compute $B^+(v, r)$.
            \item If $r > r_1$, then return ``fail''.
            \item Update $S \gets S \cup \delta^+(B^+(v, r))$, mark the nodes in $B^+(v, r)$ as bad, and remove these nodes from $G$.
        \end{enumerate}
        \item Run Cohen's algorithm (\cref{lem:cohen} with parameter $\epsilon = \frac18$) on $G$ to compute approximations $b_0(v)$ satisfying that $\tfrac78 \cdot |\Bout(v, r_0)| \leq b_0(v) \leq \tfrac98 \cdot |\Bout(v, r_0)|$. Mark all nodes $v$ satisfying that $b_0(v) < \frac78 \cdot \frac{m}{s_0}$ as bad.
    \end{enumerate}
    \item Let $R$ denote the set of remaining nodes in $G$, and return $(S, R)$.
\end{enumerate}
\end{algorithm}

In the following we formally analyze \cref{alg:cut-light}. Throughout we often condition on the event that Cohen's algorithm is correct as this happens with high probability. (Note however that we cannot necessarily detect that Cohen's algorithm has erred, and thus the algorithm may return an arbitrary answer in this rare case.)

\paragraph{Correctness of (i).}
Note that whenever the algorithm cuts a ball $\Bout(v, r)$ from the graph~$G$ and adds the boundary edges $\delta^+(\Bout(v, r))$ to $S$, it makes sure that the nodes inside and outside~$\Bout(v, r)$ lie in different strongly connected components in $G \setminus S$. As $R$ is defined to be the left-over nodes that are never cut, and as we only cut around \emph{good} nodes $v$, it suffices to prove the following claim: With high probability, at any time during the execution of the algorithm it holds that $|\Bout_G(v, r_1)| \leq \frac32 \cdot \frac{m}{s_1}$ for all good nodes $v$.

To this end, observe that we only mark nodes as good in the first step, and in this step we only mark nodes as good when $b_1(v) \leq \frac98 \cdot \frac{m}{s_1}$. By the guarantee of \cref{lem:cohen}, we therefore initially have that
\begin{equation*}
    |\Bout_G(v, r_1)| \leq \frac98 \cdot b_1(v) \leq \frac{9^2}{8^2} \cdot \frac{m}{s_1} \leq \frac{3}{2} \cdot \frac{m}{s_1}.
\end{equation*}
Finally, as we only remove nodes and edges from $G$, the size of $|\Bout_G(v, r_1)|$ can only decrease throughout the remaining execution.

\paragraph{Correctness of (ii).}
In order to prove the correctness of (ii), we crucially rely on the following claim:

\begin{claim}
When the algorithm terminates (without returning ``fail''), with high probability all nodes are marked bad.
\end{claim}
\begin{proof}
The idea is to prove that with each iteration of step 2 the number of good nodes halves. As the algorithm runs $\log n + 1$ iterations of step 2, the claim then follows.

Focus on any iteration of step 2. Let $M_0$ denote the initial set of good nodes, let $M_1, \dots, M_k$ denote the subsets of nodes that are still marked good after the respective $k = s_0 \cdot 100 \log n$ iterations of step~2.1, and let $M^*$ denote the good nodes after executing step~2.2. We clearly have that~\makebox{$M_0 \supseteq M_1 \supseteq \dots \supseteq M_k \supseteq M^*$}. Our goal is to show that $|M^*| \leq \frac12 \cdot |M_0|$. The claim is clear if~\makebox{$|M_k| \leq \frac12 \cdot |M_0|$}, so assume otherwise. We argue that after completing step~2.1, with high probability at least half of the nodes $v \in M_k$ satisfy that $|\Bout_G(v, r_0)| < \tfrac12 \cdot \tfrac{m}{s_0}$. This implies the claim as then we will mark at least half of the nodes in $M_k$ as bad while executing step 2.2. To prove this, suppose that instead more than half of the nodes $v \in M_k$ satisfy that~\smash{$|\Bout_G(v, r_0)| \geq \tfrac12 \cdot \tfrac{m}{s_0}$}. In particular, there are least $\frac{|M_k|}{2} \geq \frac{|M_0|}{4}$ such nodes. But this means that every iteration of step 2.1 \emph{succeeds} (in the sense that the algorithm samples a node $v$ passing the condition in 2.1.2) with probability at least $\frac14$. As there are $k = s_0 \cdot 100 \log n$ repetitions, with high probability at least, say,~$10 s_0$ repetitions succeed. However, every time this happens the algorithm either removes at least $\frac12 \cdot \frac{m}{s_0}$ edges from the graph or fails and stops. Thus, the total number of repetitions can be at most $2 s_0$, leading to a contradiction.
\end{proof}

With this claim in mind, we return to the correctness of Property (ii). Each node $v \in R$ has been classified as bad at some point during the execution. This could have happened in step 1 (if~\makebox{$b_1(v) > \frac98 \cdot \frac{m}{s_1}$}), or in step 2.2 (if $b_0 < \frac78 \cdot \frac{m}{s_0}$), but not in step 2.1.5 as otherwise $v$ would not have ended up in $R$. Conditioning on the correctness of Cohen's algorithm, in the former case we have that
\begin{equation*}
    |\Bout_G(v, r_1)| \geq \frac89 \cdot b_1(v) > \frac89 \cdot \frac98 \cdot \frac{m}{s_1} = \frac{m}{s_1} 
\end{equation*}
(where $G$ is the initial graph). In the latter case we have that
\begin{equation*}
    |\Bout_G(v, r_0)| \leq \frac87 \cdot b_0(v) < \frac87 \cdot \frac78 \cdot \frac{m}{s_0} = \frac{m}{s_0}
\end{equation*}
(where $G$ is the current graph), and thus also~\smash{$|\Bout_{G[R]}(v, r_0)| < \frac{m}{s_0}$} for the set $R$ returned by the algorithm. The statement~(ii) follows.

\paragraph{Correctness of (iii).}
It is clear that the source nodes of all cut edges in $S$ cannot be in~$R$. Let $e = (x, y) \in E$ denote an arbitrary edge; we show that the probability~\makebox{$\Pr(e \in S \mid x \not\in R)$} is as claimed. Note that we can only have~\makebox{$x \not\in R$} if there is some iteration of step 2 that selects a node~$v$ and samples radius $r \sim r_0 + \geom(p)$ such that~\makebox{$x \in \Bout(v, r)$}. In this iteration we might include~$e$ into $S$ if it happens that~\makebox{$e \in \delta^+(\Bout(v, r))$}. After this iteration, however, we remove the edge~$e$ from the graph and the algorithm will never attempt to include $e$ into $S$ later on. Therefore, we can upper bound the desired probability by
\begin{flalign*}
    &\Pr(e \in S \mid x \not\in R) \\
    &\qquad\qquad\leq \max_{G, v} \Pr_{r \sim r_0 + \geom(p)}(e \in \delta^+_G(\Bout_G(v, r)) \mid x \in \Bout_G(v, r)) \\
    &\qquad\qquad= \max_{G, v} \Pr_{r \sim \geom(p)}(d_G(v, y) > r_0 + r \mid d_G(v, x) \leq r_0 + r) \\
    &\qquad\qquad \leq \max_{G, v} \Pr_{r \sim \geom(p)}(d_G(v, x) + w_e > r_0 + r \mid d_G(v, x) \leq r_0 + r)
\intertext{We will now exploit the memorylessness of the geometric distribution: Conditioning on the event that $r \geq d_G(v, x) - r_0$, the random variable $r - (d_G(v, x) - r_0)$ behaves like a geometric random variable from the same original distribution (assuming that $d_G(v, x) - r_0 \geq 0$; otherwise we can simply drop the condition). Therefore:}
    &\qquad\qquad \leq \Pr_{r \sim \geom(p)}(r < w_e) \\
    &\qquad\qquad \leq p \cdot w_e.
\end{flalign*}
Recall that~\smash{$p = \frac{\ln(2s_0 / \delta)}{r_1 - r_0}$}, so the correctness follows.

\paragraph{Failure Probability.}
Next, we analyze that the algorithm returns ``fail''. Recall that this only happens if in steps 2.1.3 and 2.1.4 we sample a radius $r \sim r_0 + \geom(p)$ that satisfies~$r > r_1$. This event happens with probability
\begin{align*}
    \Pr_{r \sim r_0 + \geom(p)}(r > r_1) &= \Pr_{r \sim \geom(p)}(r > r_1 - r_0) \\
    &\leq (1 - p)^{r_1 - r_0} \\
    &\leq \exp(-p(r_1 - r_0)) \\
    &= \exp(-\ln(2s_0) / \delta) \\
    &= \frac{\delta}{2s_0}.
\end{align*}
Whenever the algorithm does not return ``fail'' in step 2.1.4, we are guaranteed to remove at least~$\frac{1}{2} \cdot \frac{m}{s_0}$ edges from the graph $G$ in step~2.1.5. In particular, there can be at most $2s_0$ repetitions of step~2.1.5 and thus at most trials of step 2.1.4. By a union bound it follows that the algorithm indeed returns ``fail'' with probability at most $\delta$.

\paragraph{Running Time.}
Finally, we analyze the running time. We run \cref{lem:cohen} once in step~1, and $\Order(\log n)$ times in step 2.2. Each call runs in time $\Order(m \log^3 n)$, so the total time spent for \cref{lem:cohen} is~$\Order(m \log^4 n)$. There are $\Order(s_0 \log^2 n)$ iterations of step 2.1. In each iteration, we spend time $\Order(1)$ for steps 2.1.1 and 2.1.4. For step 2.1.2 we spend time $\Order(\frac{m}{s_0} \log n)$---indeed, we can test whether~\smash{$|\Bout(v, r_0)| \geq \frac{1}{2} \cdot \frac{m}{s_0}$} by Dijkstra's algorithm, and if the time budget is exceeded it is clear that the bound holds. In steps 2.1.3 and 2.1.5 we spend time proportional to the size of $\Bout(v, r)$, but since we remove the nodes and edges in $\Bout(v, r)$ from $G$ afterwards in total we spend only time $\Order(m \log n)$ in these steps. Therefore, the total running time spend in step~2.1 is~\smash{$\Order(s_0 \log^2 n \cdot \frac{m}{s_0} \log n + m \log n) = \Order(m \log^3 n)$}.
\end{proof}

\subsection{Near-Linear-Time LDD} \label{sec:ldd-fast:sec:algo}
We are ready to state the LDD algorithm. Throughout this section, we fix the following parameters:
\begin{itemize}
    \item $L = \ceil{\log\log m} + 1$,
    \item \smash{$\delta = \frac{1}{\log^{10} m}$},
    \item $r_0 := 0$ and \smash{$r_\ell := r_{\ell-1} + \frac{D}{2^{L-\ell+3}} + \frac{D}{4 L}$} (for $1 \leq \ell \leq L$),
    \item \smash{$s_\ell := \min(2^{2^{L-\ell}}, m+1)$} (for $0 \leq \ell \leq L$).
\end{itemize}
With these parameters in mind, consider \cref{alg:ldd-fast}. The algorithm first deletes all edges with sufficiently large weight. Then it proceeds in $L$ iterations where in each iteration we apply \cref{lem:cut-light} to cut some edges in the graph (or reverse graph). If any of these calls returns ``fail'', then the entire algorithm restarts. In the following \cref{lem:ldd-fast-correctness,lem:ldd-fast-prob,lem:ldd-fast-time} we will show that this procedure correctly implements the algorithm claimed in \cref{thm:main-fast}.

\begin{algorithm}[t]
\caption{The near-linear-time near-optimal LDD, see \cref{thm:main-fast}.} \label{alg:ldd-fast}
\begin{enumerate}
    \item Initially let $S \subseteq E$ be the set of edges of weight at least $\frac{D}{4L}$, and remove these edges from $G$
    \item For $\ell \gets L, \dots, 1$:
    \begin{enumerate}
        \item[2.1.] Run \cref{lem:cut-light} on $G$ with parameters $\delta, r_{\ell-1}, r_\ell, s_{\ell-1}, s_{\ell}$, and let $S^+_\ell, R^+_\ell$ denote the resulting sets.
        \item[2.2.] Compute the strongly connected components in $(G \setminus S^+_\ell)[V \setminus R^+_\ell]$. Recur on each such component and add the recursively computed cut edges to $S$.
        \item[2.3.] Update $G \gets G[R^+_\ell]$ and $S \gets S \cup S^+_\ell$.
        \item[2.4.] Run \cref{lem:cut-light} on $\rev(G)$ with parameters $\delta, r_{\ell-1}, r_\ell, s_{\ell-1}, s_{\ell}$, and let $S^-_\ell, R^-_\ell$ denote the resulting sets.
        \item[2.5.] Compute the strongly connected components in $(G \setminus S^-_\ell)[V \setminus R^-_\ell]$. Recur on each such component and add the recursively computed cut edges to $S$.
        \item[2.6.] Update $G \gets G[R^-_\ell]$ and $S \gets S \cup S^-_\ell$.
    \end{enumerate}
    \item If any of the $2L$ previous calls to \cref{lem:cut-light} returns ``fail'', then we restart the entire algorithm (i.e., we reset $G$ to be the given graph, we reset $S \gets \emptyset$, and we start the execution from step~1).
\end{enumerate}
\end{algorithm}

\begin{lemma}[Correctness of \cref{alg:ldd-fast}] \label{lem:ldd-fast-correctness}
Let $u, v$ be two nodes lying in the same strongly connected component in $G \setminus S$. Then $d_G(u, v) \leq D$.
\end{lemma}
\begin{proof}
For clarity, let us denote the original (given) graph $G$ by $G_0$, and let $G$ denote the graph that is being manipulated by the algorithm. We prove the claim by induction. In the base case (when the graph has constant size) the statement can easily be enforced, so focus on the inductive case. Let $u, v$ be two arbitrary nodes.

For any iteration $\ell$, by \cref{lem:cut-light}~(i) it is clear that if $u \in R^+_\ell$ and $v \not\in R^+_\ell$ (or similarly for~$R^-_\ell$), then $u$ and $v$ cannot end up in the same strongly connected component in $G \setminus S$. If both $u, v \not\in R^+_\ell$ (or similarly for $R^-_\ell$) then the claim follows by induction since we recur on $(G \setminus S^+_\ell)[V \setminus R^+_\ell]$. This leaves as the only remaining case that $u, v \in R^+_\ell$. In particular, after completing all $L$ iterations, the only nodes left to consider are the nodes that are still contained in $G$ when the algorithm terminates. Let us call a node $v$ \emph{heavy} if
\begin{equation*}
    |\Bout_{G_0}(v, \tfrac D2)| \geq \frac{m}{2} \qquad\text{and}\qquad |\Bin_{G_0}(v, \tfrac D2)| \geq \frac{m}{2}.
\end{equation*}
In the following paragraph we prove that the only nodes $v$ remaining in $G$ are either (i) heavy or (ii) have no out- or no in-neighbors. In case (ii) clearly $v$ forms a strongly connected component in~$G \setminus S$ on its own. For (i), we claim that for any pair $u, v$ of heavy nodes it holds that $d_{G_0}(u, v) \leq D$---indeed, the out-ball~\smash{$\Bout_{G_0}(u, \frac D2)$} and the in-ball~\smash{$\Bin_{G_0}(v, \frac D2)$} necessarily intersect. It similarly holds that $d_{G_0}(v, u) \leq D$.

To prove the missing claim, let $v$ be a non-heavy node; we show that if $v$ remains in the graph~$G$ then it has no out- or no in-neighbors. By induction we show that
\begin{equation*}
    |\Bout_G(v, r_\ell)| \leq \frac{m}{s_\ell} \qquad\text{or}\qquad |\Bin_G(v, r_\ell)| \leq \frac{m}{s_\ell}.
\end{equation*}
Initially, for $\ell = L$, this is clearly true using that $v$ is not heavy, that $s_L = 2$, and that
\begin{equation*}
    r_L = \sum_{\ell=1}^L \parens*{\frac{D}{2^{L-\ell+3}} + \frac{D}{4L}} \leq \frac{D}{8} \cdot \sum_{k=0}^\infty \frac{1}{2^k} + \frac{D}{4} = \frac{D}{4} + \frac{D}{4} = \frac{D}{2}.
\end{equation*}
So now assume inductively that the claim holds for some $\ell \leq L$. Then \cref{lem:cut-light}~(ii) guarantees that
\begin{equation*}
    |\Bout_{G[R^+_\ell]}(v, r_{\ell-1})| \leq \frac{m}{s_{\ell-1}} \qquad\text{or}\qquad |\Bin_{G[R^-_\ell]}(v, r_{\ell-1})| \leq \frac{m}{s_{\ell-1}}.
\end{equation*}
Since we update $G \gets G[R^+_\ell]$ and $G \gets G[R^-_\ell]$, this is exactly as desired. In iteration $\ell = 1$ we thus have that 
\begin{equation*}
    |\Bout_G(v, r_1)| \leq \frac{m}{s_1} < 1 \qquad\text{or}\qquad |\Bin_G(v, r_1)| \leq \frac{m}{s_1} < 1,
\end{equation*}
that is, $v$ has no outgoing or incoming edges of weight less than $r_1$. Recall that initially we remove all edges of weight larger than $\frac{D}{4L} \leq r_1$, and therefore $v$ truly has no outgoing or incoming edges. This completes the proof.
\end{proof}

\begin{lemma}[Edge Cutting Probability of \cref{alg:ldd-fast}] \label{lem:ldd-fast-prob}
For any edge $e \in E$, we have
\begin{equation*}
    \Pr(e \in S) \leq \Order\parens*{\frac{w_e}{D} \cdot \log n \log\log n + \frac{1}{\poly(n)}}.
\end{equation*}
\end{lemma}
\begin{proof}
Let $e = (x, y) \in E$ be an arbitrary edge. The algorithm certainly cuts $e$ if it has weight at least $\frac{D}{4L}$, but in this case the right-hand side exceeds $1$. Otherwise, it only cuts edges via \cref{lem:cut-light} or via recursive calls. Since \cref{lem:cut-light} never cuts edges whose source lies in the returned set~$R$, each edge becomes relevant only in the first iteration $\ell$ when $x \in R_\ell^+$, or~$x \not\in R_\ell^+$ and~\makebox{$x \in R_\ell^-$}. Let us focus on the former case; the latter is analogous. Conditioning on the event that $x \in R_\ell^+$, the edge~$e$ can be cut in three ways:
\begin{enumerate}
    \item $e$ is cut by the call to \cref{lem:cut-light} (i.e., $e \in S^+_\ell$),
    \item $e$ is cut in the recursive call, or
    \item the algorithm fails and cuts $e$ during the restart.
\end{enumerate}
In particular, note that the edge cannot be cut in later iterations $\ell' < \ell$, as then we have already removed $x$ from the current graph $G$.

The probability of the first event is at most as follows, using \cref{lem:cut-light}~(iii):
\begin{align*}
    \Pr(e \in S^+_\ell \mid x \not\in R^+_\ell)
    &\leq \frac{w_e \ln(2s_{\ell-1} / \delta)}{r_\ell - r_{\ell-1}} \\
    &\leq \frac{w_e \ln(2^{2^{L - \ell + 1}} \cdot (\log m)^{10})}{\frac{D}{2^{L-\ell+3}} + \frac{D}{4L}} \\
    &\leq \frac{w_e}{D} \cdot \frac{2^{L - \ell + 1} + 10L}{\max\set*{\frac{1}{2^{L-\ell+3}}, \frac{1}{4L}}} \\
    &\leq \frac{w_e}{D} \cdot (2^{L-\ell+3} + 10L) \cdot \min\set*{2^{L-\ell+3}, 10L} \\
    &\leq \frac{w_e}{D} \cdot 2^{L-\ell} \cdot 160L.
\end{align*}
In the last step we used that $(a + b) \min\set{a, b} \leq 2 \max\set{a, b} \min\set{a, b} = 2ab$.

To analyze the events 2 and 3, let us write $p(m) = \Pr(e \in S)$ (where $m$ is the number of edges in the input). Note that each recursive call in the $\ell$-th iteration is on a graph with at most~\smash{$\frac{3}{2} \cdot \frac{m}{s_\ell} \leq \frac32 \cdot \frac{m}{2^{2^{L-\ell}}}$} edges by \cref{lem:cut-light}~(i). Therefore, the probability of event 2 is at most~\smash{$p(\frac32 \cdot \frac{m}{2^{2^{L-\ell}}})$}.

Next, focus on event 3. Each call to \cref{lem:cut-light} returns ``fail'' with probability at most~$\delta$, and by a union bound over the at most $2L$ calls the algorithm restarts with probability at most $2L\delta$. If the algorithm restarts, then clearly its randomness is independent of the previous iterations and thus it cuts $e$ with probability $p(m)$. Hence, the probability of the event 3 is at most $p(m) \cdot 2L\delta$.

By a union bound over these three probabilities we obtain that
\begin{flalign*}
    p(m) &\leq \max_\ell \parens*{\frac{w_e}{D} \cdot 2^{L-\ell} \cdot 160L + p\parens*{\frac32 \cdot \frac{m}{2^{2^{L-\ell}}}} + p(m) \cdot 2L \delta},
\end{flalign*}
or equivalently,
\begin{flalign*}
    p(m) &\leq \frac{1}{1 - 2L\delta} \cdot \max_\ell \parens*{\frac{w_e}{D} \cdot 2^{L-\ell} \cdot 160L + p\parens*{\frac32 \cdot \frac{m}{2^{2^{L-\ell}}}}} \\
    &\leq \parens*{1 + 4L\delta} \cdot \max_\ell \parens*{\frac{w_e}{D} \cdot 2^{L-\ell} \cdot 160L + p\parens*{\frac32 \cdot \frac{m}{2^{2^{L-\ell}}}}}.
\end{flalign*}
This recurrence solves to
\begin{flalign*}
    p(m) &\leq (1 + 4L\delta)^{\log_{4/3}(m)} \cdot \frac{w_e}{D} \cdot 320L \cdot \log m,
\end{flalign*}
as can easily be verified as follows; here we will bound~\smash{$\frac32 \cdot \frac{m}{2^{2^{L-\ell}}}$} on the one hand by $\leq \frac34 m$ and on the other hand by~\smash{$\leq \frac{m}{2^{2^{L-\ell-1}}}$}:
\begin{flalign*}
    p(m) &\leq \parens*{1 + 4L\delta} \cdot \max_\ell \parens*{\frac{w_e}{D} \cdot 2^{L-\ell} \cdot 160L + p\parens*{\frac32 \cdot \frac{m}{2^{2^{L-\ell}}}}} \\
    &\leq \parens*{1 + 4L\delta} \cdot \max_\ell \parens*{\frac{w_e}{D} \cdot 2^{L-\ell} \cdot 160L + (1 + 4L\delta)^{\log_{4/3}(m)-1} \cdot \frac{w_e}{D} \cdot 320L \cdot \log \frac{m}{2^{2^{L-\ell-1}}}} \\
    &\leq \parens*{1 + 4L\delta}^{\log_{4/3}(m)} \cdot \max_\ell \parens*{\frac{w_e}{D} \cdot 2^{L-\ell} \cdot 160L + \frac{w_e}{D} \cdot 320L \cdot \parens*{\log m - 2^{L-\ell-1}}} \vphantom{\parens*{\frac{m}{2^{2^{L-\ell-1}}}}} \\
    &= \parens*{1 + 4L\delta}^{\log_{4/3}(m)} \cdot \frac{w_e}{D} \cdot 320L \cdot \log m. \vphantom{\parens*{\frac{m}{2^{2^{L-\ell-1}}}}}
\end{flalign*}
Finally, we plug in the chosen values for $L \leq \Order(\log\log m)$ and $\delta = (\log m)^{-10}$ to obtain that
\begin{align*}
    p(m) &\leq \parens*{1 + \frac{\Order(\log\log m)}{(\log m)^{10}}}^{\log_{4/3}(m)} \cdot \Order\parens*{\frac{w_e}{D} \cdot \log m \log\log m} \\
    &\leq \exp\parens*{\frac{\Order(\log\log m)}{(\log m)^{10}} \cdot \log_{4/3}(m)} \cdot \Order\parens*{\frac{w_e}{D} \cdot \log m \log\log m} \\
    &\leq \Order\parens*{\frac{w_e}{D} \cdot \log m \log\log m},
\end{align*}
as claimed. A final detail is that all previous bounds condition on the high-probability event that none of the calls to \cref{lem:cut-light} returns an invalid answer (but not ``fail''). To take this into account, the edge cutting probability increases additively by $\frac{1}{\poly(n)}$ (for an arbitrarily large polynomial).
\end{proof}

\begin{lemma}[Running Time of \cref{alg:ldd-fast}] \label{lem:ldd-fast-time}
The expected running time is $\Order(m \log^5 n \log\log n)$.
\end{lemma}
\begin{proof}
Ignoring the cost of recursive calls, one execution of \cref{alg:ldd-fast} involves $\Order(L)$ calls to \cref{lem:cut-light} each of which runs in time $\Order(m \log^4 n)$. (The additional computations such as computing strongly connected components runs in linear time.) To take the recursive calls and restarts into account, let~$T(m)$ denote the running time given $m$ edges. We recurse on subgraphs with, say,~$m_1, \dots, m_r$ edges where $m_i \leq \frac34 m$ and $\sum_{i=1}^r m_i \leq m$. Moreover, the probability that the algorithm restarts is at most $2L \delta \leq (\log m)^{-9}$. Thus, we can bound
\begin{equation*}
    T(m) = \sum_{i=1}^r T(m_i) + 2L\delta T(m) + \Order(m \log^4 n \log\log n),
\end{equation*}
which, similarly to the analysis in \cref{lem:ldd-fast-prob}, can be solved to $T(m) \leq \Order(m \log^5 n \log\log n)$.
\end{proof}

We remark that we have not attempted to optimize the log-factors in the running time here and it is likely that the overhead can be reduced. An obvious improvement to shave one log-factor from the running time is to employ a faster priority queue in Dijkstra's algorithm, e.g.\ as developed by Thorup~\cite{Thorup03}.
\bibliographystyle{plainurl}
\bibliography{ref}

\end{document}